\documentclass[12pt]{article}

\newcommand{\blind}{1}

\usepackage[toc,page]{appendix}

\RequirePackage{amsthm,amsmath,amsfonts,amssymb}

\usepackage[usenames,dvipsnames]{xcolor}
\usepackage{tcolorbox}
\usepackage{tabularx}
\usepackage{colortbl}
\tcbuselibrary{skins}

\newcolumntype{Y}{>{\raggedleft\arraybackslash}X}

\tcbset{tab1/.style={fonttitle=\bfseries\large,
mytitle/.style={title={Example~\thetcbcounter\ifstrempty{#1}{}{: #1}}},
fontupper=\normalsize\sffamily,
colback=yellow!10!white,colframe=red!75!black,colbacktitle=Salmon!40!white,
coltitle=black,center title,freelance,frame code={
\foreach \n in {north east,north west,south east,south west}
{\path [fill=red!75!black] (interior.\n) circle (3mm); };},}}

\tcbset{tab2/.style={enhanced,fonttitle=\bfseries,fontupper=\normalsize\sffamily,
colback=yellow!10!white,colframe=red!50!black,colbacktitle=Salmon!40!white,
coltitle=black,center title}}

\usepackage{array}
\newcolumntype{Y}{>{\centering\arraybackslash}X}

\RequirePackage[authoryear]{natbib}
\RequirePackage[colorlinks,citecolor=blue,urlcolor=blue]{hyperref}
\RequirePackage{graphicx}
\usepackage{sidecap}
\usepackage{multirow}
\usepackage{float}
\usepackage{mathtools}
\usepackage{color}
\usepackage{xfrac}
\usepackage{bigints}
\usepackage{subfigure}
\usepackage{bbm}
\usepackage{booktabs}
\usepackage{siunitx, tabularx}
\usepackage{adjustbox}
\usepackage{xr}
\usepackage{arydshln,,leftidx}
\usepackage{verbatim}
\usepackage{ upgreek }
\usepackage{algorithm,algpseudocode}
\usepackage{amssymb}
\usepackage{epstopdf}
\usepackage{bm}
\usepackage{bigints}

\newtheorem{proposition}{Proposition}[section]
\newtheorem{theorem}{Theorem}[section]
\newtheorem{lemma}[theorem]{Lemma}
\newtheorem{corollary}{Corollary}
\newtheorem{remark}{Remark}

\newcommand{\h}[1]{\widehat{#1}}

\newcommand{\grad}{{\nabla}}

\newcommand{\real}{\mathbb{R}}

\newcommand{\cE}{{\mathcal{E}}}

\newcommand{\cG}{{\mathcal{G}}}
\newcommand{\cS}{{\mathcal{S}}}
\newcommand{\cA}{{\mathcal{A}}}
\newcommand{\cU}{{\mathcal{U}}}
\newcommand{\cO}{{\mathcal{O}}}

\newcommand{\rZ}{{\mathrm{z}}}

\newcommand{\tp}{\intercal}

\newcommand{\twofigs}[2]{
\hbox to\hsize{\hss
\vbox{\psfig{figure=#1,width=2.7in,height=2.0in}}\qquad
\vbox{\psfig{figure=#2,width=2.7in,height=2.0in}}
\hss}}

\makeatletter
\newcommand\currentlabel{\@currentlabel}
\makeatother

\begin{document}

\date{}

\def\spacingset#1{\renewcommand{\baselinestretch}%
{#1}\small\normalsize} \spacingset{1}


\if1\blind
{
  \title{\bf Exact Selective Inference with Randomization}
  \author{Snigdha Panigrahi\thanks{
    The author acknowledges support by NSF-DMS 1951980 and NSF-DMS 2113342.}\hspace{.2cm} \\
    Department of Statistics,
		University of Michigan,
         MI, USA.\\
          and \\
    Kevin Fry\thanks{
    The author acknowledges support by NSF GRFP.}\hspace{.2cm} \\
    Department of Statistics,
		 Stanford University,
         CA, USA.\\
    and \\
    Jonathan Taylor\thanks{
    The author acknowledges support in part by ARO grant 70940MA.}\hspace{.2cm} \\
    Department of Statistics,
		 Stanford University,
         CA, USA.}
  \maketitle
} \fi

\if0\blind
{
  \bigskip
  \bigskip
  \bigskip
  \begin{center}
    {\LARGE\bf Exact Selective Inference with Randomization}
\end{center}
  \medskip
} \fi

\begin{abstract}
We introduce a pivot for exact selective inference with randomization.
Not only does our pivot lead to exact inference in Gaussian regression models, but it is also available in closed form.
We reduce the problem of exact selective inference to a bivariate truncated Gaussian distribution.
By doing so, we give up some power that is achieved with approximate maximum likelihood estimation in \cite{panigrahi2019approximate}.
Yet our pivot always produces narrower confidence intervals than a closely related data splitting procedure.
We investigate the trade-off between power and exact selective inference on simulated datasets and an HIV drug resistance dataset.
\end{abstract}

\noindent%
\vfill

\newpage
\spacingset{1.3}

\section{Introduction}

The polyhedral method by \cite{exact_lasso} introduced confidence intervals for exact selective inference in Gaussian regression models.
This method provides valid inferences for selected parameters by conditioning on the outcome of selection.
A pivot is obtained for each selected parameter from a truncated Gaussian distribution, provided the outcome of selection can be described by linear constraints, also known as polyhedral constraints. 
However, as shown by \cite{kivaranovic2018expected}, confidence intervals based on this pivot can have infinite length in expectation.

Randomizing data at the time of selection and conditioning on the outcome of randomized selection produces narrower confidence intervals than the polyhedral method.
\cite{kivaranovic2020tight} formally establish that some of these randomized procedures guarantee intervals with bounded lengths.
A stumbling block for subsequent inference, however, is the lack of a pivot in closed form after marginalizing over the added randomization variables.
For example, the pivot based on randomized response, as in \cite{randomized_response}, or on data carving, which involves holding out a random subsample during selection, as in \cite{optimal_inference}, cannot be directly computed.

Recent work by \cite{panigrahi2019approximate} bypassed this computational hurdle by proposing an approximate Gaussian pivot through maximum likelihood estimation. 
The approximate pivot is obtained by solving a convex optimization problem which yields the selection-adjusted maximum likelihood estimator (MLE) and observed Fisher information matrix.
Although computationally appealing, this pivot may not provide adequate coverage if the approximation is inaccurate.  
Moreover, it can be difficult to determine the reliability of the approximation in practical settings. 
Inaccuracies can arise when the dimensions of the problem are significantly larger than the number of available samples. 
To provide an example, consider the case where $n=500$ independent and identically distributed (i.i.d.) samples are generated from a Gaussian linear regression model with $p=1000$ predictors, of which $25$ are true signals with magnitude of $\sqrt{2t \log p}$ and the rest are noise.
We conduct $500$ rounds of simulations with $t$ taking values $0.5$, $0.75$ and $1$.
In all three scenarios, the coverage probability of the approximate pivot produced by ``MLE'' is below the target level $0.90$, as reported under ``Coverage" in Table \ref{Tcov}.
\bigskip

\newcounter{mycounter}

\noindent 
\begin{tcolorbox}[step=mycounter, label=Tcov, width=.475\textwidth, nobeforeafter, title=Table 1. Coverage, tab2,tabularx={X||Y ||Y}]
$t$ & MLE    & Exact          \\\hline\hline
$0.5$   & $86.09\%$ & $89.94\%$  \\\hline
$0.75$ & $86.69\%$ & $90.02\%$  \\\hline
$1$  & $85.56\%$ & $90.30\%$ \\\hline
\end{tcolorbox}
\hfill
\begin{tcolorbox}[step=mycounter, label=Tlen, width=.475\textwidth, nobeforeafter, title=Table 2. Length, tab2,tabularx={X||Y ||Y}]
$t$ & MLE    & Exact        \\\hline\hline
$0.5$   & $20.83$ & $27.14$  \\\hline
$0.75$ & $20.86$ & $26.94$ \\\hline
$1$  & $21.09$ & $27.44$  \\\hline
\end{tcolorbox}
\bigskip

In this paper, we offer a new pivot for selective inference with randomization.
We aim at exact selective inference in closed form, without requiring a case-by-case treatment for different models.
In exchange, we give up some power that is achieved with the approximate Gaussian pivot in \cite{panigrahi2019approximate}.
We note this trade-off between the coverage probabilities and the averaged lengths of the intervals for both methods, ``MLE'' and ``Exact'' (our proposed method), in Tables \ref{Tcov} and \ref{Tlen}.
Despite sacrificing some power, our pivot produces more reliable inferences that (roughly) attain the target coverage probability $0.90$ in all three scenarios. 

We structure the remaining paper as follows. 
In Section \ref{sec:background}, we review existing work in selective inference and present a toy example as a warm-up to our method. 
Our main result in Section \ref{sec:main} presents a pivot for exact selective inference in the well known setting of LASSO regression.
In Section \ref{sec:examples}, we show that our pivot readily applies to different instances of selective inference in Gaussian regression models.
In Section \ref{sec:sims}, we investigate the quality of selective inference with our new pivot on simulated datasets.  
We apply our method to a publicly available HIV drug resistance dataset in Section \ref{sec:real}.
Our empirical experiments investigate the price for exact inference—in terms of power—on simulated and real datasets.
A brief discussion in Section \ref{sec:conclusion} concludes.
Proofs for our results are deferred to the Appendix.

\section{Background}
\label{sec:background}

\subsection{Some preliminaries}

We begin by fixing a set of notations that are used throughout the paper.
Let $[d]=\{1,2,\cdots, d\}$ for $d\in \mathbb{N}$. 
The symbol $e_j  \in \real^d$ is understood as a vector with $1$ in the $j^{\text{th}}$ entry and $0$ elsewhere.
For $\eta \in \real^d$ and $\Theta \in \real^{d\times d}$, $\eta_{j}= e_j^{\tp}\eta$ is the $j^{\text{th}}$ entry of $\eta$, $\Theta_{j,k}= e_j^{\tp} \Theta e_k$ is the $(j,k)^{\text{th}}$ entry of $\Theta$, and $\Theta_{[j]}$ is the $j^{\text{th}}$ row of $\Theta$.

We use $\phi(x; \theta, \Theta)$ to denote the density function of a Gaussian variable with the mean vector $\theta \in \real^d$ and covariance matrix $\Theta\in \real^{d\times d}$ at $x$. 
In particular, when $d=1$, $\theta=0$, $\Theta=1$, we let $\phi(x)$ be the density of a standard normal variable and let $\Phi(x)$ be its cumulative distribution function.
Denote by
$$\text{TP}^{[a,b]}(\theta, \vartheta)= \Phi\left(\frac{1}{\vartheta}(b- \theta)\right)- \Phi\left(\frac{1}{\vartheta}(a-\theta)\right)$$
the truncation probability that a univariate Gaussian variable with mean $\theta$ and variance $\vartheta^2$ lies in the interval $[a,b]$.

For background on selective inference, we consider the standard setting of LASSO regression with a fixed design matrix.
Suppose that we have a vector of outcomes $y \sim \mathcal{N}(\mu, \sigma^2 I_n)\in \real^n$ for an unknown mean parameter $\mu$ and a matrix of $p$ fixed features $X\in \real^{n \times p}$.
We observe $w \sim \mathcal{N}(0_p, \Omega)$, a $p$-dimensional randomization variable that is drawn independently of $y$.
Consider solving 
\begin{equation}
\widehat{b} = \underset{b \in \real^p}{\text{argmin}}\ \frac{1}{2} \|y - Xb\|_2^2  + \frac{\epsilon}{2} \|b\|_2^2  + \lambda\|b\|_1 -w^{\tp} b
\label{randomized:LASSO}   
\end{equation}
with regularization parameter $\lambda \in \real^{+}$.

The selection algorithm in \eqref{randomized:LASSO} gives us a noisy version of the LASSO, which is called the randomized LASSO in \cite{harris2016selective}.
A small, fixed value of $\epsilon \in \real^{+}$ in the objective of the randomized LASSO simply ensures us the existence of a solution. 
The variance of the Gaussian randomization variable is a tuning parameter that is similar to the split proportion in data splitting. 
It lets us control how much information we use to select a model versus how much we use for inference. 
As an example, consider $\Omega(\tau^2) = \tau^2 I_p$. 
If we increase the value of $\tau^2$, it means that we perform a noisier model selection, which reserves more information for inference.
Later in the paper, we discuss incorporating a Gaussian randomization scheme that is related to data splitting.

After solving \eqref{randomized:LASSO}, we seek inference for a set of post-selection parameters.
Here is a common example.
Let 
$$
E= \left\{j\in [p]: |\text{sign}(\widehat{b}_j)|= 1\right\}.
$$
Having observed the selected subset of features $E= \cE$, we infer for 
\begin{equation}
\beta^{\cE} = (X_{\cE}^{\tp} X_{\cE})^{-1} X_{\cE}^{\tp} \mu \in \real^{|\cE|},
\label{post:LASSO:target}
\end{equation}
which is the best linear representation of $\mu$ using the selected subset of features $X_{\cE}$.
For brevity, let $c^{j}= X_{\cE} (X_{\cE}^{\tp} X_{\cE})^{-1}e_j \in \real^{n}$ for $j\in [|\cE|]$.
This allows us to write each entry of $\beta^{\cE}$ as
$$\beta^{\cE}_j = (c^{j})^{\tp} \mu.$$
Note that $\beta^{\cE}_j$ depends on $y$ and $w$ through $c^{j}$, which in turn depends on $\cE$.

\subsection{Existing work}

We begin by reviewing two existing methods that are closely related to our current proposal.
The first method offers an exact pivot for selective inference when solving the standard version of LASSO, without randomization. 
The second method provides an approximate pivot after solving the randomized LASSO.

Both pivots are obtained from a conditional distribution of the outcome variable after conditioning on a proper subset of the observed event.
Conditioning on $\{E=\cE\}$ is ideal if we wanted inference for $\beta^{\cE}$.
However, the ideal event is usually complicated to describe in terms of $y$ and $w$, making the conditional distribution of $y$ given $\{E=\cE\}$ 
less amenable to inferences.
Therefore, conditioning on a subset of the selection event that has a simpler description is a practical solution, which can ensure valid and feasible selective inference.

\noindent\textbf{The polyhedral method}.
Consider solving the standard LASSO \citep{tibs_lasso}, which involves setting $\epsilon=0$ and $w=0_p$ in the objective of \eqref{randomized:LASSO}.
We denote the set of selected features as $E_0$. 
Note that we distinguish $E_0$ from the selected set $E$, which is obtained from solving the randomized LASSO.

Having observed $E_0=\cE_0$, fix $c_0^{j}= X_{\cE_0} (X_{\cE_0}^{\tp} X_{\cE_0})^{-1}e_j \in \real^{n}$ that leads to
$$\beta^{\cE_0}_j = (c_0^{j})^{\tp} \mu,$$
our parameters after selection.
Let $S_0\in \real^{|E_0|}$ be the vector of nonzero signs.
Let $\h{\beta}^{\cE_0}$ denote the least squares estimator when we regress $y$ against $X_{\cE_0}$ and
let 
$$
\h{\Gamma}_0^{j}= \left(I-\frac{c_0^{j}(c_0^{j})^{\tp}}{\|c_0^{j}\|^2_2}\right)y
$$
be the projection of $y$ onto the orthogonal complement of the subspace spanned by $c_0^{j}$.

Conditional on $E_0=\cE_0$, $S_0=\cS_0$ and the value of $\h{\Gamma}^{j}_0$, the polyhedral method in \cite{exact_lasso} gives an exact pivot by truncating a univariate Gaussian variable, with mean $\beta^{\cE_0}_{j}$ and variance $\sigma^2\|c_0^{j}\|_2^2$, to an interval $[H_{-}^{j}, H_{+}^{j}]$.
The pivot takes the form
\begin{equation}
\mathcal{P}^j_{\text{Poly}}(\beta^{\cE_0}_j)=\dfrac{\bigintsss_{-\infty}^{\h{\beta}^{\cE_0}_j}\phi\left((\sigma \|c^{j}_0\|_2)^{-1}(x-\beta^{\cE_0}_j)\right) \cdot 1_{[H_{-}^{j}, H_{+}^{j}]}(x) dx}{\bigintsss_{-\infty}^{\infty}\phi\left((\sigma \|c^{j}_0\|_2)^{-1}(x-\beta^{\cE_0}_j)\right) \cdot 1_{[H_{-}^{j}, H_{+}^{j}]}(x) dx},
\label{Lee:pivot}
\end{equation}
where the expressions for $H_{-}^j$ and $H_{+}^j$ depend on $\cE_0$, $\cS_0$ and $\h{\Gamma}_0^{j}$.

\noindent\textbf{The MLE method}.  
Next, we turn to selective inference with the randomized LASSO.
The approximate MLE method in \cite{panigrahi2019approximate} uses the likelihood of $y$ when conditioned on 
\begin{equation}
\{G= \cG\},
\label{cond:MLE}
\end{equation}
where 
$$G=\partial_{\; \widehat{b}\;}{ \|b\|_1}$$
is the subgradient of the $\ell_1$-penalty at the randomized LASSO solution. 
Similar to the polyhedral method, the conditioning event is a proper subset of the ideal event $\{E= \cE\}$.

Let $\h{b}^{\cE}$ and $\h{I}^{\cE}$ denote the MLE and the observed Fisher information matrix in this conditional likelihood.
An approximate Gaussian pivot for $\beta^{\cE}_{j}$ is given by
\begin{equation}
\mathcal{P}^j_{\text{MLE}}(\beta^{\cE}_j)= {\Phi}\left(\frac{1}{\sqrt{(\h{I}^\cE)^{-1}_{j,j}}} (\h{b}^{\cE}_j-\beta^{\cE}_j)\right).
\label{approx:Gaussian:pivot}
\end{equation}
Equivalently, confidence intervals for each component of $\beta^{\cE}$ are calculated by centering them around the $j^{\text{th}}$ entry of the MLE, with the variance estimated by the corresponding diagonal entry of the observed Fisher information matrix. 
However, the seemingly simple Gaussian pivot involves computing the exact conditional likelihood function, which cannot be done in closed form, hence making it difficult to compute the two estimators. 
To overcome this, the approximate MLE method derives approximate values for  $\h{b}^{\cE}$ and $\h{I}^{\cE}$, which rely on a consistent approximation to the exact conditional likelihood. 

\subsection{Warm-up}

We can now informally present the central idea of our paper using a toy example with two features (p=2).
We solve \eqref{randomized:LASSO} with $\epsilon =0$ and $w \sim \mathcal{N}(0_2, \Omega) \in \real^2$, where $\Omega= \tau^2 X^T X$.

Say that we select the full model, i.e., $\cE=\{1,2\}$, and that we focus on the first component of the $2$-dimensional post-selection parameter 
$$\beta_1^{\cE}= (c^{1})^{\tp} \mu.$$
Let $\h{\beta}^{\cE}$ be the least squares estimator when regressing $y$ against $X_{\cE}$ and let $\widehat\beta_1^{\cE}$ be its first component.

Introducing some additional notation, let $O\in \mathbb{R}^2$ denote the non-zero randomized LASSO solution in this example.
Let $S=\text{sign}(O)$ be the corresponding sign vector and let $\cS$ be the observed value of $S$.
Recall that the existing MLE method makes inferences after conditioning on the event $\{G=\cG\}$.
Since $\cG= \cS$ in this example, it is easy to see that this conditioning event can be described as 
 \begin{equation}
 \left\{-\text{diag}(\cS) O<0_2\right\}.
 \label{cond:toy}
 \end{equation}
While the previously mentioned MLE method obtains an approximate Gaussian pivot with this conditioning event, we can simplify the event by conditioning on some additional information that reduces the conditioning event to an interval on the real line.
This is the central idea behind constructing an exact pivot in closed form.

 In this specific toy example, we condition on
 $$A= O_2- \dfrac{e_1^\tp (X^\tp X)^{-1} e_2}{e_1^\tp (X^\tp X)^{-1} e_1}O_1,$$
 in addition to conditioning on the value of $G$. 
Because
$$O= \left(\dfrac{1}{e_1^\tp (X^\tp X)^{-1} e_1} (X^\tp X)^{-1} e_1\right) O_1 + \begin{pmatrix} 0 \\ A\end{pmatrix},$$
the initial conditioning event in \eqref{cond:toy} simplifies to
$$\left\{I^1_{-} \leq O_1 \leq I^1_{+}\right\}$$
after conditioning on $A$, where $[I^1_{-},I^1_{+}]$ is a fixed interval.
Consequently, we can obtain an exact pivot in closed form by computing the joint bivariate distribution of $\widehat\beta_1^{\cE}$ and $O_1$ when truncated to the region $\real \times [I^1_{-},I^1_{+}]$.
We explain our choice for additional conditioning and derive a bivariate truncated Gaussian distribution in the next section.


To conclude, we note the difference between our method and the polyhedral method, which is shown in Figure \ref{fig:newpivot}.
For drawing selective inference, the polyhedral method truncates the Gaussian distribution of $\widehat\beta_1^{\cE}$ to the interval $[H_{-}^{1}, H_{+}^{1}]$, while our method truncates the joint bivariate distribution of $\widehat\beta_1^{\cE}$ and $O_1$ to $\real \times [I^1_{-},I^1_{+}]$.
\begin{figure}[h]
\begin{center}
\centerline{\includegraphics[width=1\linewidth]{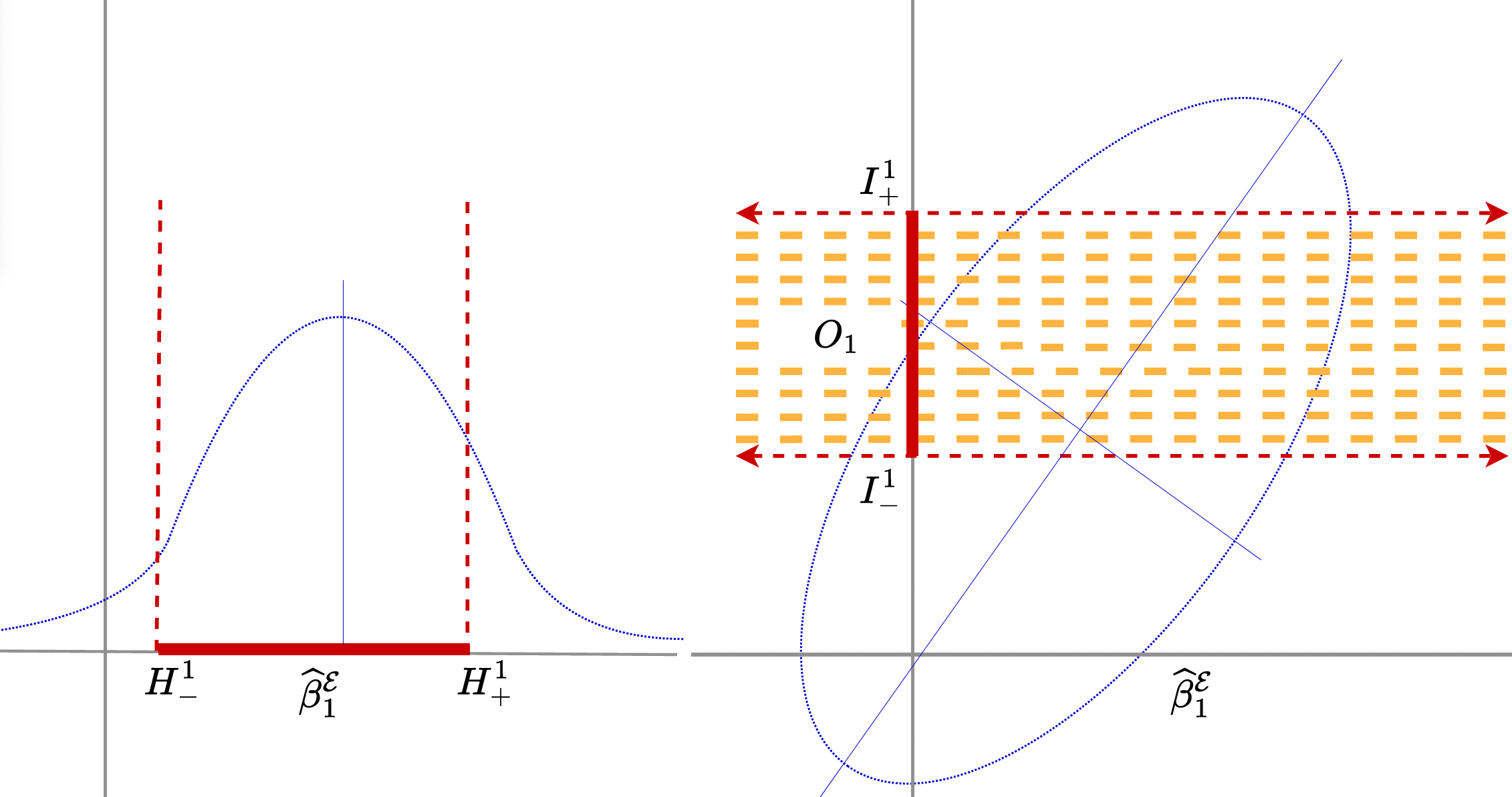}}
\end{center}
\vspace{-1cm}
\caption{Comparison with the polyhedral method.
}
\label{fig:newpivot}
\end{figure} 

\subsection{Connections with other work}

Several papers have demonstrated the effectiveness of the conditional approach for selective inference across various problems, as evidenced by \cite{lee2014exact, yang2016selective, suzumura2017selective, charkhi2018asymptotic, hyun2018exact, chen2020valid, zhao2019selective, tanizaki2020computing, gao2020selective, duy2020computing}.
A significant focus of the current research in this field is on enhancing the power of earlier approaches. 
Before we discuss these improvements, it is worth noting that two other approaches to selective inference have been studied in parallel.

The first is the simultaneous inference approach, which has been investigated in \cite{posi, bachoc2020uniformly}.
This approach is not customized to a particular selection method, but the downside is that the confidence intervals are relatively long and may not permit easy calculations in some instances. 
The second approach is data splitting. 
This method allows for valid selective inference when the available data can be split into two independent sets. 
One set is used as training data for the selection process, while the other set is held out as validation data for selective inference. 
Combined with the bootstrap in regression models, \cite{rinaldo2019bootstrapping} conduct selective inference by splitting the sample space.
Recently, new forms of data splitting have been introduced by \cite{rasines2021splitting, leiner2021data, neufeld2022inference}, which split each observation into two parts to construct a training set for selection and a validation set for selective inference. 
However, these variants of data splitting lose power by discarding data used in selection. 
In our simulations, we confirm that inverting our pivot results in narrower confidence intervals than two such forms of data splitting.

There are two main branches of the conditional approach that have improved power and overcome the limitations of the polyhedral approach. 
The first branch of work involves choosing a minimal conditioning set that can be achieved in some special settings. 
For example, \cite{liu2018more} condition on strictly less information than the polyhedral method when inference is based on a full linear model $y\sim \mathcal{N}(X\beta, \sigma^2 I_n)$. 
In the saturated model $y\sim \mathcal{N}(\mu, \sigma^2 I_n)$, \cite{le2022more} apply parametric programming to avoid conditioning on the signs of the LASSO coefficients and \cite{carrington2023improving} condition on less information to provide inference for detected changepoints.
The second branch of work utilizes randomization variables at the time of selection to remedy a loss in power. 
Some of these randomized procedures can be viewed as a more efficient alternative to data splitting and appear as data carving in existing literature \citep{optimal_inference, panigrahi2019carving, schultheiss2021multicarving}.
Randomization variables have been used to deliver powerful Bayesian inference after model selection in papers by \cite{selective_bayesian, panigrahi2020approximate, panigrahi2022radiogenomics}. 
Our work falls in the latter category, where we provide a principled approach to choose a conditioning event and construct a pivot thereof that can work with different Gaussian regression models after selection.

It is not a new idea to find a conditioning event that can lead to a bivariate truncated distribution. 
In \cite{kivaranovic2020tight}, one such construction is noted, where noise is added to a Gaussian response as proposed by \cite{randomized_response}. 
This work achieved an exact pivot by conditioning on the projection of the noisy response onto the orthogonal complement of the subspace spanned by the direction vector of interest. 
However, in our paper, we have employed a different randomization scheme, which involves adding noise to the optimization objective. 
As demonstrated in \cite{huang2023selective}, this scheme has the potential to be applied to a broad range of M-estimation problems, not just the least squares estimation problem. 
For example, while adding Gaussian noise to a binary response in logistic regression might not be meaningful, adding noise to the log-likelihood would create a noisy estimation problem. 
Although our primary focus in this paper is on exact selective inference, our pivot is likely to generalize and provide asymptotic inference in a more comprehensive context. 
We provide a discussion on this topic in our concluding remarks.

\section{Exact selective inference with the LASSO}
\label{sec:main}

\subsection{Conditioning event}
\label{sec3.1}

We continue using the randomized LASSO to explain our approach in the general Gaussian regression setting. 

Defining some notations, we denote by $O \in \real^{|\cE|}$ the active (nonzero) LASSO solution, and by $S=\text{sign}(O)$ the associated sign vector. 
Throughout, we assume that the active components of the LASSO solution are stacked before its inactive components. 
The $p$-dimensional subgradient of the $\ell_1$-penalty at the randomized LASSO solution is denoted by
$$G = \begin{pmatrix} S \\ U \end{pmatrix},$$
where $U \in \real^{p-|\cE|}$ collects the components of the subgradient subvector in $\cE^c$. 
To represent the realized values of the variables $O$, $S$, and $U$, we use the symbols $\cO$, $\cS$, and $\cU$, respectively.

At the randomized LASSO solution, observe that
\begin{equation*}
w= Py+  Q\cO + R \cU +T,
\end{equation*}
where
\begin{align}
\begin{gathered}
P=-\begin{bmatrix}X_{\cE}^{\tp}\\ X_{\cE^c}^{\tp} \end{bmatrix}, \; Q= \begin{bmatrix} X_{\cE}^{\tp} X_{\cE} + \epsilon I_{|\cE|} \\ X_{\cE^c}^{\tp} X_{\cE} \end{bmatrix}, \;R= \begin{bmatrix} 0_{|\cE|, p-|\cE|} \\ \lambda I_{p-|\cE|} \end{bmatrix}, \; T= \begin{pmatrix} \lambda \cS \\ 0_{p-|\cE|}\end{pmatrix},
\end{gathered}
\label{PQRT}
\end{align}
and we have assumed that the active components are stacked before the inactive ones in our matrices.

As demonstrated in the previous section's toy example, we first identify a conditioning event that will guide us to a pivot for exact selective inference.
Extending the method by \cite{panigrahi2019approximate}, we condition on
$$\{G= \cG\},$$
which can be described as
\begin{equation}
\{L O< M, \ U=\cU\},
\label{cond:event}
\end{equation}
for
$$
L= -\text{diag}(\cS),\ M= 0_{|\cE|}.
$$

To reduce our conditioning event to an interval and obtain a closed-form pivot, we condition on some more information.
Proposition \ref{Lem:cond} states this event, which is equivalent to truncating a linear combination of $O$ to a fixed interval.
To present this result, we introduce a few matrices that rely on the covariance of the randomization variables and the matrices defined in \eqref{PQRT}.
Let
\begin{align*}
\begin{gathered}
\Theta = (Q^{\tp}\Omega^{-1} Q)^{-1}, \ P^{j} = \frac{1}{\|c^{j}\|^2_2}Pc^{j}\\
r^{j} =  Q^{\tp} \Omega^{-1} P^{j} \in \real^{|\cE|},\ Q^j=\frac{1}{(r^{j})^{\tp} \Theta r^{j}}\Theta r^{j},
\end{gathered}
\end{align*}
 for $j\in [|\cE|]$.

\begin{proposition}
Define the variables
\begin{equation}
A^{r^j} = \left( I_{|\cE|}-  Q^j(r^{j})^{\tp} \right) O \in \real^{|\cE|}.
\label{additional:cond}
\end{equation}
For $j\in [|\cE|]$, it holds that
$$
\left\{ G= \cG, A^{r^j}= \cA^{r^j} \right\} = \left\{I^j_{-} <  (r^{j})^{\tp} O < I^j_{+}, \; U=\cU, A^{r^j}= \cA^{r^j} \right\},
$$
where 
$$
I^j_{-} = \underset{k\in S^j_{-}}{\text{max}} \left\{\frac{1}{L_{[k]}^{\tp}Q^j } (M_k - L_{[k]}^{\tp}\cA^{r^j})\right\},\  I^j_{+} = \underset{k\in S^j_{+}}{\text{min}}\left\{ \frac{1}{L_{[k]}^{\tp}Q^{j} } (M_k - L_{[k]}^{\tp}\cA^{r^j})\right\},
$$
and
$$S^j_{-} = \left\{k :  L_{[k]}^{\tp} \Theta r^{j} <0\right\},  \ S^j_{+} = \left\{k :  L_{[k]}^{\tp} \Theta r^{j} >0\right\}.$$
\label{Lem:cond}
\end{proposition}

From the previous result, we observe that the conditioning event involves extra information in the form of $A^{r^j}$, representing linear combinations of the active LASSO coefficients.  
We motivate our choice of conditioning event later. 
In the next section, we obtain a pivot for $\beta_j^{\cE}$ by conditioning on the event in Proposition \ref{Lem:cond}.

\subsection{Pivot}
\label{sec3.2}

Let $\h{\beta}^{\cE}$ be the least squares estimator obtained by regressing $y$ on $X_{\cE}$.
Specifically, let $\h{\beta}^{\cE}_j=(c^{j})^\tp y$ denote the $j^{\text{th}}$ entry of $\h{\beta}^{\cE}$. 
Define $\h{\Gamma}^{j}$ as the projection of $y$ onto the orthogonal complement of the subspace spanned by $c^{j}$. 

Note that 
$$\mu = \frac{c^{j}}{\|c^{j}\|^2_2} (c^{j})^{\intercal} \mu + \mathcal{P}^{\perp}_{c^{j}} \mu.$$
When inferring for $(c^{j})^{\intercal} \mu$, the projection $\mathcal{P}^{\perp}_{c^{j}} \mu$ includes nuisance parameters.
To eliminate these parameters, we follow a similar approach as in \cite{exact_lasso} and condition on $\h{\Gamma}^{j}$.
This allows us to obtain a conditional density that involves only our parameter of interest, $\beta_j^{\cE}$.
We can then use its CDF to obtain a pivot.

To state our main result, we introduce the functions
\begin{align*}
\begin{gathered}
\Lambda(y, \cU)= -(P^{j})^{\tp} \Omega^{-1}(Py+ R \cU + T),\\
\Delta(y, \cU)= -\Theta Q^{\tp}\Omega^{-1}(Py+ R \cU + T).
\end{gathered}
\end{align*} 

\begin{theorem}
Define the random variable 
\begin{equation*}
\begin{aligned} 
\mathcal{P}^j_{\text{\normalfont Exact}} (\beta^{\cE}_j)=\dfrac{\bigintsss_{-\infty}^{\h{\beta}^{\cE}_j}\phi\left(\frac{1}{\sigma^j}(x-\lambda^j\beta^{\cE}_j- \zeta^j)\right) \cdot \text{\normalfont TP}^{[I^j_{-},I^j_{+}]}(\theta^j(x), \vartheta^j)dx }{\bigintsss_{-\infty}^{\infty} \phi\left(\frac{1}{\sigma^j}(x-\lambda^j\beta^{\cE}_j- \zeta^j)\right)\cdot \text{\normalfont TP}^{[I^j_{-},I^j_{+}]}(\theta^j(x), \vartheta^j)dx},
\end{aligned}
\end{equation*}
where the constants $\vartheta^j$, $\sigma^j$, $\lambda^j$, $\zeta^j$ and the univariate function $\theta^j$ are computed as
\begin{align*}
\begin{gathered}
 (\vartheta^j)^2=(r^{j})^{\tp}\Theta r^{j}, \ (\sigma^{j})^2= \left(\frac{1}{\sigma^2 \|c^{j}\|_2^2} + (P^{j})^{\tp}\Omega^{-1} P^{j} -(\vartheta^j)^2\right)^{-1},\\
 \lambda^j=\dfrac{1}{\sigma^2\|c^{j}\|_2^2}(\sigma^{j})^2,\  \zeta^j=(\sigma^{j})^2\cdot\left(\Lambda(\h{\Gamma}^{j}, \cU) - (r^{j})^{\tp} \Delta(\h{\Gamma}^{j}, \cU)\right).\\
 \theta^j(x)= (r^{j})^{\tp}\Delta(\h{\Gamma}^{j}, \cU)-(\vartheta^j)^2 x.
\end{gathered}
\end{align*}
Conditioned on the event in Proposition \ref{Lem:cond}, $\mathcal{P}^j_{\text{\normalfont Exact}} (\beta^{\cE}_j)$ is distributed as a $\text{\normalfont Unif}(0,1)$ variable.
\label{thm:main}
\end{theorem}

Inverting the pivot in Theorem \ref{thm:main} gives a confidence interval for $\beta^{\cE}_j$. 
At a predetermined significance level $\alpha$, a two-sided confidence interval for $\beta^{\cE}_j$ is equal to
$$
 \left(L^{j}_{\alpha}, U^{j}_{\alpha}\right)= \left\{b \in \real:  \mathcal{P}^j_{\text{\normalfont Exact}} (b) \in \left[\frac{\alpha}{2},1-\frac{\alpha}{2} \right]\right\}.
$$

A few comments are in order here.
\begin{remark}
The choice of the simple model $y \sim \mathcal{N}(\mu, \sigma^2 I_n)\in \real^n$ was made for ease of presentation. 
However, it is essential to note that our pivot can also be applied to other Gaussian regression models, such as the model in \cite{optimal_inference} where $y\sim \mathcal{N}(X_{\cE} \beta_{\cE}, \sigma^2 I_n)$ or the full model in \cite{liu2018more} where $y\sim \mathcal{N}(X\beta, \sigma^2 I_n)$. 
The only difference would be in the definition of $c^j$ for each model, which depends on the post-selection parameters chosen for inference.
\end{remark}

\begin{remark}
\cite{liu2018more} noted that the ideal conditioning event could vary across different models.
In some special situations, such as when inferring for the selected regression parameters in a full model $y \sim \mathcal{N}(X\beta, \sigma^2 I_n)$, conditioning on less information than the polyhedral method is possible. 
However, in our method, the conditioning event based on the outcome of the randomized selection algorithm is the same for different regression models. 
Therefore, the construct of our pivot is consistent regardless of our modeling preferences.
In our empirical experiments, we demonstrate the performance of our pivot in the selected and full models.
\end{remark}


\subsection{Pivot motivated by data carving}
\label{sec3.2}

We instantiate our pivot using a Gaussian randomization scheme that can be seen related to the data carving proposal in \cite{optimal_inference}.
Data carving is similar to data splitting in that it involves using a subset of the data for selection, but differs from data splitting in that it uses the entire dataset for inference instead of relying solely on the held-out portion. 

Suppose that we apply the LASSO method to a subsample of size $n_1$ drawn from a dataset that contains $n$ i.i.d. pairs of observations $(y_i, x_i) \in \mathbb{R}^{p+1}$. 
Then, the LASSO on the subsample is asymptotically equivalent to solving a randomized LASSO with 
\begin{equation}
w\sim N\left(0_p, \tau^2\mathbb{E}[x_1 x_1^\tp]\right),
    \label{randomization:carving}
\end{equation}
where $\tau^2= \sigma^2 \cdot\frac{(n-n_1)}{n_1}$. 
This result is formally stated in \cite{selective_bayesian}.
We provide some additional details in the Appendix to offer insights into this connection.
This motivates us to solve
\begin{equation}
   \underset{b \in \real^p}{\text{minimize}}\ \frac{1}{2}\|y-Xb\|_2^2 + \lambda \|b\|_1 -w^{\tp} b,
\label{randomized:LASSO:eg}   
\end{equation}
with $w$ drawn from a Gaussian distribution with mean $0_p$ and covariance $\tau^2 X^\tp X$, which is the sample analog of the covariance matrix in \eqref{randomization:carving}.
Recall that this was also the randomization scheme in our toy example.

Using this particular form of Gaussian randomization, we can observe that the value of $r^j$ is directly proportional to $e_j\in \real^{|\cE|}$. 
This means that our conditioning event is equivalent to truncating the $j^{\text{th}}$ active LASSO coefficient $O_j$ to an interval on the real line, which is depicted in Figure \ref{fig:newpivot}.
As a result, our pivot in Theorem \ref{thm:main} simplifies as follows.
\begin{corollary}
\label{exact:pivot:carving}
Suppose that $\Omega$ is defined according to \eqref{randomization:carving}.
Then,
$$\mathcal{P}^j_{\text{\normalfont Exact}} (\beta^{\cE}_j)=\dfrac{\bigintsss_{-\infty}^{\h{\beta}^{\cE}_j}\phi\left((\sigma \|c^{j}\|_2)^{-1}(x-\beta^{\cE}_j)\right) \cdot \text{\normalfont TP}^{[I^j_{-},I^j_{+}]}(\theta^j(x), \vartheta^j)dx }{\bigintsss_{-\infty}^{\infty} \phi\left((\sigma \|c^{j}\|_2)^{-1}(x-\beta^{\cE}_j)\right)\cdot \text{\normalfont TP}^{[I^j_{-},I^j_{+}]}(\theta^j(x), \vartheta^j) dx},
$$
where 
$$ (\vartheta^j)^2= \frac{1}{\tau^2 \|c^{j}\|^2_2},$$
and $\theta^j:\real\to \real$ is equal to
$$\theta^j(x)=  \frac{1}{\tau^2\|c^{j}\|^2_2} \left(\lambda e_j^{\tp}(X_{\cE}^{\tp} X_{\cE})^{-1} \cS -x\right).$$
\end{corollary}

Upon revisiting our toy example, 
we recall that the polyhedral method truncates the Gaussian distribution of $\widehat{\beta}^{\cE}_j$ to the interval 
$[H_{-}^{j}, H_{+}^{j}]$ for $j\in \{1,2\}$.
In contrast, our new pivot replaces the indicator function $1_{[H_{-}^j, H_{+}^j]}(x)$ with the Gaussian probability
$$
\text{\normalfont TP}^{[I^j_{-},I^j_{+}]}(\theta^j(x), \vartheta^j)
$$
in the integrand of \eqref{Lee:pivot}.

\begin{remark}
Of course, solving \eqref{randomized:LASSO:eg}    is not exactly the same as applying the LASSO on a subsample of size $n_1$. 
If selection is carried out on a randomly selected subsample, then our pivot would provide asymptotic selective inference rather than exact, due to the asymptotic equivalence between the Gaussian randomization and selection on the subsample.
Since our current focus is on providing exact guarantees for selective inference, we defer a formal proof of this to future work.
\end{remark}
 
\subsection{Choice of conditioning}
\label{sec3.3}

We come back to our conditioning event in Proposition \ref{Lem:cond}.

Denote by 
\begin{equation}
\left(L_{\alpha}^{j,\cG}, U_{\alpha}^{j,\cG}\right)
\label{CI:z}
\end{equation}
the confidence interval for $\beta^{\cE}_j$ if we had based inference on the conditional distribution of $\h{\beta}^{\cE}$, given the event $\{G=\cG\}$ as done by the MLE method.
In principle, we can fix any arbitrary vector $\eta\in \real^{|\cE|}$ and further condition on 
\begin{equation}
A^{\eta}= \left(I -\frac{1}{\eta^{\tp}\Theta \eta}\Theta \eta \eta^{\tp}\right)O.
\label{A:eta}
\end{equation}

By writing 
$$
O = \frac{\Theta\eta}{\eta^{\tp}\Theta \eta} \eta^{\tp} O + A^{\eta},
$$
our conditioning event simplifies to an interval as:
\begin{equation*}
\begin{aligned}
\left\{ G=\cG, A^{\eta}=\cA^{\eta}\right\} &= \left\{ LO < M, U=\cU,  A^{\eta}=\cA^{\eta}\right\} \\
&= \left\{I^{\eta}_{-} <  \eta^{\tp} O < I^{\eta}_{+}, U=\cU,  A^{\eta}=\cA^{\eta}\right\},
\end{aligned}
\end{equation*}
where $I^{\eta}_{-}$ and $I^{\eta}_{+}$ now depend on $L$, $M$, $\Theta$, and $\cA^{\eta}$.

If we follow the same steps as before, then we can obtain an exact pivot for $\beta^{\cE}_j$ by using a truncated distribution that is supported on $\real\times [I^{\eta}_{-},I^{\eta}_{+}]$. 
If we let $\eta=r^j$, it leads to the conditioning event in Proposition \ref{Lem:cond} and to the proposed pivot.

Now we address choosing $\eta$, which determines the additional conditioning information.
Consider a situation when selection has no impact, i.e., the truncated distribution is no different from the usual distribution with no further adjustment for selection. 
Our specific choice $\eta=r^{j}$ is motivated from the fact that no extra price is paid by conditioning on $A^{r^j}$ in the situation described above.
In other words, the confidence intervals produced by our pivot
$$\left\{ \left(L^{j}_{\alpha}, U^{j}_{\alpha}\right): j\in \cE \right\}$$
narrow down to the intervals in \eqref{CI:z} as selection has a diminishing impact.
We formalize this fact in Proposition \ref{choice:cond}.

\begin{proposition}
Let $A^{\eta}$ be defined according to \eqref{A:eta}.
Then, we have 
\begin{equation*}
\begin{aligned}
\text{\normalfont Var}\left(\h{\beta}_j^{\cE} \ \Big\lvert \ U= \cU, A^{r^j}= \cA^{r^j}, \h{\Gamma}^j= g\right)&= \text{\normalfont Var}\left(\h{\beta}_j^{\cE} \ \Big\lvert \ U= \cU, \h{\Gamma}^j= g\right)\\
&= \underset{\eta}{\text{maximum}} \ \text{\normalfont Var}\left(\h{\beta}_j^{\cE}\ \Big\lvert \ U= \cU, A^{\eta}= \cA^{\eta}, \h{\Gamma}^j= g\right).
\end{aligned}
\end{equation*}
\label{choice:cond} 
\end{proposition}

It is worth noting that there might be other ways to choose the direction $\eta$. 
One such option is to choose $\eta$ in a way that minimizes the variance of the  bivariate truncated distribution that arises when we condition on $\left\{ G=\cG, A^{\eta}=\cA^{\eta}\right\}$. 
While this approach seems ideal, it is not straightforward as the resulting optimization is not convex in $\eta$ and may not be easily solvable. 

Another option is to condition on all active LASSO coefficients, except for the $j^{\text{th}}$ one, when inferring the effect of the $j^{\text{th}}$ selected variable.
 However, this choice will not generalize well to other models post selection. 
 For example, if we add a new variable $X^{*}$ to the selected model and fit it using the features $E\cup \{X^*\}$, it is unclear what to condition on when inferring for the effect of $X^{*}$ in this selected model.
    
In contrast, our approach to choosing $\eta$ is simple yet principled, which applies broadly to Gaussian linear models with our form of additive randomization introduced at the selection step.

\section{More examples}
\label{sec:examples}

\subsection{General setup}
The randomized LASSO serves as our first concrete instance in the paper.
Our pivot serves as a generally applicable framework to other Gaussian regression models. 
We introduce the general setup and provide further examples in this section.

Suppose that we solve
\begin{equation}
    \underset{b \in \real^p}{\text{minimize}}\ \  \ell(b; d) + P_{\lambda}(b) - w^{\tp}b,
    \label{randomized:alg}
\end{equation}
for $w\sim \mathcal{N}(0, \Omega)$.
For example, letting $d= (y, X)$, and fixing  
$$\ell(b; y, X)= \frac{1}{2}\|y-Xb\|^2_2 + \frac{\epsilon}{2}\|b\|_2^2, \text{ and } P_{\lambda}(b)= \lambda \|b\|_1$$
gives us the randomized LASSO in Section \ref{sec:main}. 

As before, we begin by conditioning on a proper subset of the event 
$\{E=\cE\}$, which we denote by $\{G=\cG\}$.
Denote the KKT conditions of stationarity for \eqref{randomized:alg} by
\begin{equation}
    w= \grad \ell(V; d) + \partial P_{\lambda}(V)
    \label{KKT}
\end{equation}
where 
$$V=\begin{pmatrix} O \\ U \end{pmatrix} \in \real^p$$
represent $p$ optimization variables at the solution of the randomized selection algorithm.

Our general setup for selective inference relies on two basic assumptions.
First, the stationarity conditions in equation \eqref{KKT} can be represented as:
\begin{equation}
w = P\widehat{d} + Q\cO + R\cU +T,
\label{KKT:affine}
\end{equation}
where $\widehat{d}$ is a statistic based on the data $d$. 
In other words, the stationarity conditions admit a linear representation in our optimization variables.
Second, we assume that the conditioning event can be expressed as:
\begin{equation}
\left\{ G = \cG \right\} = \left\{ LO < M, \ U = \cU \right\},
\label{cond:affine}
\end{equation}
meaning that the event $\left\{ G = \cG \right\}$ is equivalent to imposing linear constraints on our optimization variables.

We use the linear representations in equations \eqref{KKT:affine} and \eqref{cond:affine} to construct our pivot after conditioning on the additional information $A^{r^j}=\cA^{r^j}$.
We turn to more examples for illustrating our pivot.

\subsection{Revisiting the randomized LASSO}
We provide an alternate pivot which begins by conditioning on the event considered by \cite{exact_lasso}.
That is, let 
\begin{equation}
\left\{G=\cG\right\}= \left\{E=\cE, S =\cS\right\},
\label{Lee:cond}
\end{equation}
which means that we condition on the set of selected features along with the signs of their corresponding LASSO coefficients.
We note that this event can be represented as
$$
\{LV<M\},
$$
where 
$$
V= \begin{pmatrix} O \\ U\end{pmatrix} \in \real^p, \ L=\begin{bmatrix} -\text{diag}(\cS) & 0_{|\cE|, p-|\cE|}\\ 0_{p-|\cE|, |\cE|} & I_{p-|\cE|}\\  0_{p-|\cE|, |\cE|} & -I_{p-|\cE|}\end{bmatrix},\  M=\begin{pmatrix}0_{|\cE|}\\ 1_{p-|\cE|} \\ 1_{p-|\cE|} \end{pmatrix}.
$$
Let $\cO \in \real^p$ be the realized value of $V$.

The KKT conditions of stationarity in this example are given by
\begin{equation*}
w= Py+  Q\cO +T,
\end{equation*}
for
$$P=-X^{\tp}, \; Q= \begin{bmatrix} X_{\cE}^{\tp} X_{\cE} + \epsilon I_{|\cE|} & 0_{|\cE|, p-|\cE|} \\ X_{\cE^c}^{\tp} X_{\cE} & I_{p-|\cE|} \end{bmatrix}, \; T= \begin{pmatrix} \lambda \cS \\ 0_{p-|\cE|}\end{pmatrix}.$$
Clearly, the two linear representations in \eqref{KKT:affine} and  \eqref{cond:affine} are met. 

Proceeding as before, we condition further on $A^{r^j}$.
This allows us to reduce our conditioning event to a single linear constraint in our optimization variables $V \in \real^p$.
The main difference with the pivot in Section \ref{sec:main} is that we start with a different conditioning event, which matches with the conditioning event used in the polyhedral method.

\subsection{Randomized screening of correlations}

Suppose that we screen features using their marginal correlations with the outcome.
For a fixed threshold $\lambda\in \real^{+}$, a randomized screening procedure selects features which satisfy
$$|X_j^{\tp} y + w_j| >\lambda,$$
for $w\in \mathcal{N}(0_p, \Omega)$.
Let $E$ denote the set of selected features.
Equivalently, this selection can be written as 
\begin{equation}
   \underset{b \in \real^p}{\text{minimize}}\ \frac{1}{2}\|b-X^{\tp}y\|_2^2 + \chi_{K_{\lambda}}(b) -w^{\tp} b,
\label{randomized:screening:eg}   
\end{equation}
where 
$$K_{\lambda}=\{o: \|o\|_{\infty} < \lambda\}, \text{ and } \chi_{K_{\lambda}}(b)= 
\begin{cases} 
      0 & \text{ if } b\in K_{\lambda} \\
      \infty &  \text{otherwise}. 
\end{cases}
$$

We define our optimization variables as follows. 
First, let $O\in \real^{|E|}$ collect the active components of the subgradient for the penalty at the solution. 
Define $U= |X_{E^c}^{\tp} y + w_{E^c}| \in \real^{p-|E|}$ and $S=\text{Sign}(X_{E}^{\tp} y + w_{E})$.
Consider the conditioning event
$$
\{G=\cG\} =\{ E =\cE, S= \cS, U= \cU\}.
$$
We note that the representation in \eqref{KKT:affine} is satisfied with 
$$P=-X^{\tp}, \; Q= \begin{bmatrix} I_{|\cE|} \\ 0_{p-|\cE|, |\cE|}  \end{bmatrix}, \; R= \begin{bmatrix} 0_{|\cE|, p-|\cE|} \\ I_{p-|\cE|} \end{bmatrix},\; T= \begin{pmatrix} \lambda \cS \\ 0_{p-|\cE|}\end{pmatrix},$$
for $\widehat{d}= y$.
It is also easy to see that the event $\{G=\cG\}$ satisfies \eqref{cond:affine} by letting
$$
L= -\text{diag}(\cS),\ M= 0_{|\cE|}.
$$
We condition on the event in Proposition \ref{Lem:cond} as before and obtain our pivot for exact selective inference.

\subsection{Randomized SLOPE}

In this example, we consider solving a randomized version of the SLOPE algorithm in \cite{bogdan2015slope}.
The randomized SLOPE, for $w\sim \mathcal{N}(0, \Omega)$, is given by
\begin{equation}
   \underset{b \in \real^p}{\text{minimize}}\ \frac{1}{2}\|y-Xb\|_2^2 + \sum_{j=1}^p\lambda_j |b|_{[j]} -w^{\tp} b,
\label{randomized:SLOPE:eg}   
\end{equation}
where
$$|b|_{[1]} \geq |b|_{[2]} \geq \cdots \geq |b|_{[p]}$$
denote the magnitudes (absolute values) of entries of $b$ in decreasing order.
For simplicity sake, we assume that the $p$ tuning parameters $\lambda_j$ are unique, and let
$$
\Lambda = \begin{pmatrix} \lambda_1 & \cdots & \lambda_p \end{pmatrix}. 
$$

Fixing some notations, we let $O\in \real^q$ collect the magnitudes of the distinct, nonzero components of the SLOPE solution.
Without losing generality, we let
$$
O_1 > O_2 >\cdots>O_q. 
$$
The collection of selected features with an estimated SLOPE coefficient equal to $O_k$, in magnitude, is denoted by $C_k$.
Let the indices of the features in $C_k$ be $I_k\subseteq [p]$, and let the size of this collection be equal to $|C_k|$.
Denote by $C_0$ the collection of features which are not selected by the randomized SLOPE, and let $I_0$ be their corresponding indices.
For $k\in 0\cup [q]$, let $U'_k \in \real^{|C_k|}$ collect the entries of the subgradient for the penalty that are present in the set $I_k$.
For $k\geq 1$, we drop the smallest component of $U'_k$ which we denote by $s'_k$, and call the resulting subvector $U_k\in \real^{|C_k|-1}$.
Then, let 
$$U\in \real^{p-q} = \begin{pmatrix}  U_1 \\ \vdots \\ U_k \\ U_0' \end{pmatrix}, \ S'\in \real^q = \begin{pmatrix}  s'_1 \\ \vdots \\ s'_q  \end{pmatrix}$$
At last, we let the signs of the selected features in $C_k$ be $S_k$, and then fix 
$\bar{X}_k = \sum_{j\in C_k} \text{diag}(S_k) X_k$, and let
$$
X_0 = \begin{bmatrix} \bar{X}_1 & \cdots \bar{X}_q\end{bmatrix}.
$$

With these notations, it is easy to note that the representation in \eqref{KKT:affine} holds with 
$$P=-X^{\tp}, \; Q=X^{\tp} X_0, \;R= \begin{bmatrix} 0_{q, p-q} \\ I_{p-q} \end{bmatrix}, \; T= \begin{pmatrix} S' \\ 0_{p-q} \end{pmatrix}$$
for $\widehat{d}=y$.
Now if $G$ is the subgradient of the SLOPE penalty, then we observe that \eqref{cond:affine} is satisfied by letting $L\in \real^{q\times q}$ be a matrix of all zeroes except for the entries
$$L_{i,i}=-1, \ L_{i,i+1}=1, \text{ for } i\in [q-1], \text{ and } L_{q,q}=-1,$$ and 
$M= 0_{q}$.

\subsection{Selective reporting with bootstrapped data}

Bootstrapping is a commonly used statistical technique that estimates the variance of an estimator or prediction error. 
It is also used to construct confidence intervals for unknown parameters. 
However, researchers may choose to report inferences for only a selected subset of these parameters, based on the magnitude of their estimators or their statistical significance. 
This is known as selective reporting and can invalidate inferences for the chosen parameters.

To address this issue, we use the framework of (approximate) penalized Gaussian regression to cast the problem of selective reporting after bootstrapping. 
With this approach, we can easily apply our pivot to conduct selective inference with our bootstrapped samples. 
We assume that $\widehat\beta\in \mathbb{R}^p$ is our estimator and $\beta \in \mathbb{R}^p$ is the vector of unknown parameters.
Let 
$$\left\{\widehat\beta^{(b)}: b\in [B]\right\}$$ 
be the collection of estimators that are re-computed on $B$ bootstrapped samples.
We assume that $\widehat\beta -\beta \approx \mathcal{N}(0_p, \Sigma)$, and bootstrapping works in the sense that
$$ \widehat\beta^{(b)} - \widehat\beta \approx \widehat\beta -\beta$$
in distribution, for $b\in [B]$, and that the two variables are independent.
Formally, this implies that we can approximate the distributions of $\widehat\beta -\beta$ and $\widehat\beta^{(b)} - \widehat\beta$ by independent Gaussian distributions, with mean $0_p$ and covariance $\Sigma$.
Hereafter, we use the above-stated Gaussian distribution to model the two variables.

We begin by considering two independent estimators
$$\widetilde\beta = \widehat\beta + \alpha (\widehat\beta^{(b)} - \widehat\beta ), \ \widetilde\beta^\perp = \widehat\beta -\frac{1}{\alpha} (\widehat\beta^{(b)} - \widehat\beta )$$
for $\alpha>0$ and a randomly chosen bootstrapped estimator $\widehat\beta^{(b)}$ for $b \in [B]$.
One may now select a subset of relevant parameters to report using $\widetilde\beta$, and use the independent estimator $\widetilde\beta^\perp$ for selective inference.
Note that $\widehat\beta$ can be substituted with an estimator based on the bootstrapped samples in case there is no direct access to $\widehat\beta$. 
Similarly,  the Gaussian covariance matrix $\Sigma$ can be replaced with an estimator for the variance of $ \widehat\beta$.

The strategy of dividing an initial estimator $\widehat\beta$ into two independent estimators for $\beta$ can be viewed as being related to the proposals by \cite{rasines2021splitting} and \cite{leiner2021data}. 
We notice that $\widetilde\beta$ equals $\widehat\beta^{(b)}$ when $\alpha$ is equal to 1, meaning that we use a single bootstrapped sample in this case. 
By adopting the splitting method, we carry out selective inference using $\widetilde\beta^\perp = 2\widehat\beta - \widehat\beta^{(b)}$.

In our example, we propose to select a subset of parameters by solving
 \begin{equation}
        \underset{b\in \mathbb{R}^p}{\text{minimize}} \; \frac{1}{2} (\widetilde\beta - b)^\tp \Sigma^{-1} (\widetilde\beta - b) + \lambda \|b\|_1.
         \label{eq:bootlasso}
  \end{equation}
We observe that \eqref{eq:bootlasso} is equivalent to the optimization
$$
\underset{b\in \mathbb{R}^p}{\text{minimize}} \; \frac{1}{2} b^{\tp} \Sigma^{-1} b - b^\tp \Sigma^{-1}\widehat{\beta}  + \lambda \|b\|_1 - w^\tp b,
$$
where $w= \alpha \Sigma^{-1}  (\widehat\beta^{(b)} - \widehat\beta ) \sim N(0, \alpha^2 \Sigma^{-1})$.
Equivalently, selective reporting with an $\ell_1$-penalty can be recognized as the randomized LASSO in \eqref{randomized:LASSO} by letting
$$y= \Sigma^{-1/2}\widehat\beta \in \mathbb{R}^p,\; X = \Sigma^{-1/2} \in \mathbb{R}^{p\times p},\text { and } \epsilon=0.$$
In order to construct selective inference, we can make use of our pivot, which was previously described in the paper.
Note that our pivot will not only use $\widetilde\beta^\perp$, but also the estimator $\widetilde\beta$ to make inferences for the selected or reported entries of $\beta$.

\section{Simulations}
\label{sec:sims}

\subsection{Settings and modeling strategies}
To evaluate how well our pivot performs, we use data generated from a sparse Gaussian model given by:
\begin{equation}
y = X_{E^*}\beta_{E^*} + \epsilon.
\label{true:model}
\end{equation}
Here, $\epsilon \in\mathbb{R}^n$ is a vector of i.i.d Gaussian errors with mean $0$ and variance $\sigma^2$ and $E^* \subset [p]$ is a sparse support set for $\beta \in \mathbb{R}^p$.

We construct the feature matrix $X$ by drawing $n = 500$ samples from a $p=200$ dimensional Gaussian distribution $\mathcal{N} \left( 0_{p}, \Sigma \right)$ 
with 
$$
\Sigma_{ij}= 0.9^{ \lvert i -j \rvert}.
$$
Then, we simulate $y$ from the model in \eqref{true:model} with noise level $\sigma^2=3$ and $|E^*|$=5. 

We design two main settings to study how our method compares with previously proposed procedures in selective inference.
In our first setting, we vary the proportion of data used for model selection, also called ``Split Proportion".
We compare methods that use roughly the same amount of information for feature selection as data splitting at a prespecified value of split proportion.
We elaborate on this further when we describe the different methods under study.
In the second setting, we vary the signal strength of the non-zero entries of $\beta$ to investigate how different methods compare under varying signal regimes. 
Specifically, we set the magnitude of the nonzero entries for $\beta$ as $\sqrt{2f\log p}$.
We vary the fraction f in the set $\{0.50, 1, 1.5, 2, 3\}$, and number the corresponding settings as ``Signal Regimes $1-5$" in our plots.

In each setting, we consider two common modeling strategies.
\begin{enumerate}
    \item Full Model: \quad we model our response using the full set of features.
    In other words, we model our response as
    $$y \sim \mathcal{N}(X\beta, \sigma^2 I_n).$$
    We estimate the noise level in our data by using residuals based on a regression of $y$ against all $p$ features.  
        
  In each round of simulation, we select a sparse set of features  $E=\cE$. 
  We then consider inference for the selected coefficients in the full model. To be precise, our parameters, after selection, are:    
  $$
    \beta^{\cE} = (\beta_j: j\in \cE)^{\tp} \in \mathbb{R}^{|\cE|}.
    $$
     The vector $\beta^{\cE}$ contains entries of $\beta$ that are present in the selected set $\cE$.
  \item Selected Model: \quad we model our response using the selected set of features $\cE$.
  That is, we use the model
    $$y \sim \mathcal{N}(X_{\cE}\beta_{\cE}, \sigma^2 I_n).$$
    In this case, we estimate the noise level by using the residuals based on a regression of our response against the selected features.
   
    We infer for the partial regression coefficients in the selected model
    $$
    \beta^{\cE} =  (X_{\cE}^{\tp} X_{\cE})^{-1} X_{\cE}^{\tp} X_{E^*}\beta_{E^*} \in \real^{|\cE|},
    $$
    that are obtained by projecting the true mean $X_{E^*}\beta_{E^*}$ onto the subspace spanned by the selected features.
\end{enumerate}
In both models, we adopt a plug-in approach to estimate the noise variance. 
We comment on this approach below.
\begin{remark} 
The work by \cite{randomized_response} supports the use of a plug-in estimator for $\sigma$ as long as it is an consistent estimator of the true noise variance before selection.  
While we use the parametric form of the fitted model to obtain a plug-in estimator, the plug-in approach can be more general in principle. 
For example, one can estimate the noise variance by using nonparametric function estimation methods, which separates the task of error estimation from the precise parametric modeling of our response.
\end{remark}

Our reported findings are based on $500$ rounds of simulations for each pair of setting and modeling strategy.

\subsection{Methods}

We compare the following methods: 
\begin{enumerate}
    \item ``Exact": \quad our current method to conduct exact selective inference with Gaussian randomization after solving \eqref{randomized:LASSO};
    \item ``MLE": \quad the approximate maximum likelihood method reviewed in Section \ref{sec:background}; this method conducts selective inference with an approximate Gaussian pivot after selecting features through  \eqref{randomized:LASSO};
    \item ``Polyhedral+": \quad this method, introduced in \cite{liu2018more}, applies the standard LASSO algorithm for selecting features and then conducts inference for the selected coefficients in the Full Model by conditioning on strictly less information than the polyhedral method in \cite{exact_lasso};
    \item ``Split": \quad this method is based on data splitting; we divide the training data into two independent parts, using $n_1$ samples for model selection with the standard LASSO algorithm, which is followed by using the remaining samples for valid selective inference.
    \item ``UV": \quad this method uses the UV decomposition proposed by \cite{rasines2021splitting}, where selection and inference are conducted on two independent datasets; for a randomization variable $\tilde{w}\sim N(0_n, \sigma^2 f I_n)$, selection is conducted with the U-estimator $Y+ \tilde{w}$ as the train response and selective inference is conducted for the selected parameters using the V-estimator $Y-\frac{1}{f}\tilde{w}$ as the test response. 
\end{enumerate}

The two methods ``Exact" and ``MLE" are constructed under the Gaussian randomization scheme that was discussed in Section \ref{sec3.3}.
Specifically, we fix the randomization covariance as $\Omega = \tau^2 X^{\tp} X$
with
$$ 
\tau^2= \widehat{\sigma}^2 \cdot\frac{(n-n_1)}{n_1},
$$
where $\hat\sigma$ is the estimated noise level in our model.
Both these methods are compared to data splitting which uses $n_1$ samples for feature selection.
To implement the ``UV" method, we replace $\sigma$ by its estimated value under our model and to ensure fair comparisons, we set $f=\frac{n-n_1}{n_1}$ in our analysis.
We report comparsions of our method with the ``UV" method across different signal regimes.

\begin{remark}
In our simulations, we choose not to use the polyhedral method from \cite{exact_lasso}. 
This is because, on average across $500$ simulations, the interval lengths it produces are much longer than the other four methods we are using. 
In fact, the polyhedral method returns infinitely long interval estimates in every setting, which is consistent with the findings in \cite{kivaranovic2018expected}. 
\end{remark}

\begin{remark}
When considering the Full Model, we include a summary of the performance of the ``Polyhedral+'' method, along with the two randomized methods "Exact" and "MLE". 
The benefits of utilizing the entire dataset rather than dividing it into samples are significantly noticeable when applying the Full Model. 
Therefore, we exclude the split-based methods from our summary plots since they produce considerably longer intervals, on average, compared to the other methods.

For the Selected Model, we compare the four randomized methods used in our simulations. 
We note that the "Polyhedral+" method is designed to provide selective inference only under the Full Model and does not apply to the Selected Model.
\end{remark}

\subsection{Findings}

First, we evaluate the accuracy of feature selection by using 
$$\text{F}1 \; \text{score}= \dfrac{\text{true positives}}{\text{true positives} + \dfrac{1}{2}(\text{false positives} +\text{false negatives}) }$$
in our two main settings.

In the left panel of Figure \ref{fig:1}, we vary the split proportion $\rho= \frac{n_1}{n}$ at a fixed strength of signals while keeping the signal strength fixed. 
We use the randomized LASSO method to conduct feature selection with Gaussian randomization that corresponds to the prespecified split proportion $\rho$; 
``Exact" and ``MLE" provide inference for the effects of the features selected with this randomized version of the LASSO.
The standard implementation of the LASSO, which is used by "Polyhedral+", applies feature selection on the entire dataset, and is represented in the plot as ``Standard".
We note that the distribution of the $\text{F}1 \; \text{score}$ for the Gaussian randomization scheme closely resembles the randomization involved in the related ``Split" procedure. 
As expected, the accuracy of selection increases with higher values of split proportion, eventually matching the accuracy attained by ``Standard" on the full data.

In the right panel of Figure \ref{fig:1}, we fix the  split proportion at $0.80$ and vary our signal regimes in the set $1-5$. 
Consistent with expectations, the accuracy of feature selection increases as we strengthen the signals. 
Notably, all the methods used for feature selection perform almost equally well at a split proportion of $0.80$, which is consistent with the findings of the left panel of the plot.

\begin{figure}[h]
\begin{center}
\centerline{\includegraphics[width=1.\linewidth]{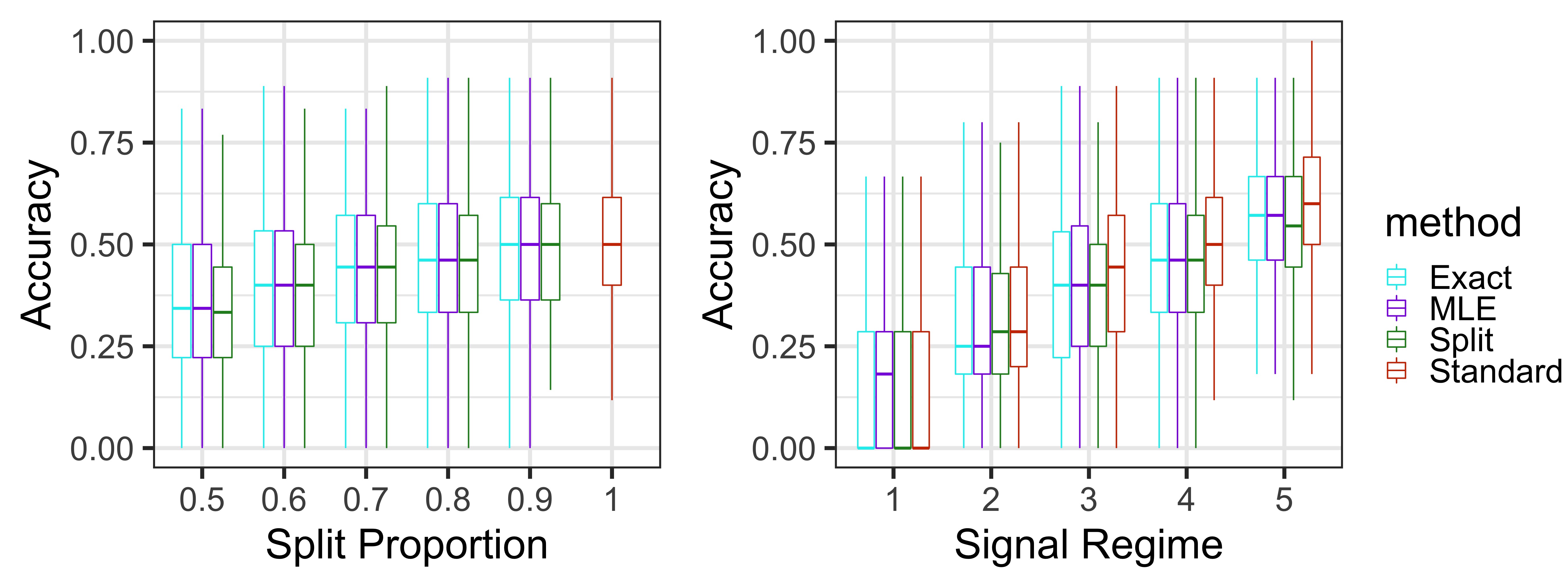}}
\end{center}
\vspace{-1cm}
\caption{Accuracy based on quality of feature selection. Left panel shows distribution of $\text{F}1 \; \text{score}$ at fixed Signal regime 3 as split proportion $\rho$ varies. Right panel shows distribution of $\text{F}1 \; \text{score}$ at fixed split proportion $\rho=0.80$, as signal regimes vary from $1-5$.}
\label{fig:1}
\end{figure} 

Next, we compute the false coverage rate of the confidence intervals for different methods, which is equal to
$$
\text{FCR}=  \dfrac{\left\{ j \in \cE: \beta^\cE_j \not\in \mathcal{C}^\cE_j\right\}}{\max(|\cE|,1)}. 
$$
In Figures \ref{fig:2} and \ref{fig:3}, we plot the coverage rates $1-\text{FCR}$ for $90\%$ confidence intervals under the two models, Full Model and Selected Model. 
The averaged coverage rate, over all replications, is highlighted by the dot mark.
The horizontal broken line at $0.90$ depicts the target coverage rate for all the methods.

We note that ``Exact" achieves the desired rate of coverage as do the previous methods of selective inference.  
This pattern remains consistent even as we change the split proportion or the strength of signals in different signal regimes.

\begin{figure}[H]
\begin{center}
\centerline{\includegraphics[width=1.\linewidth]{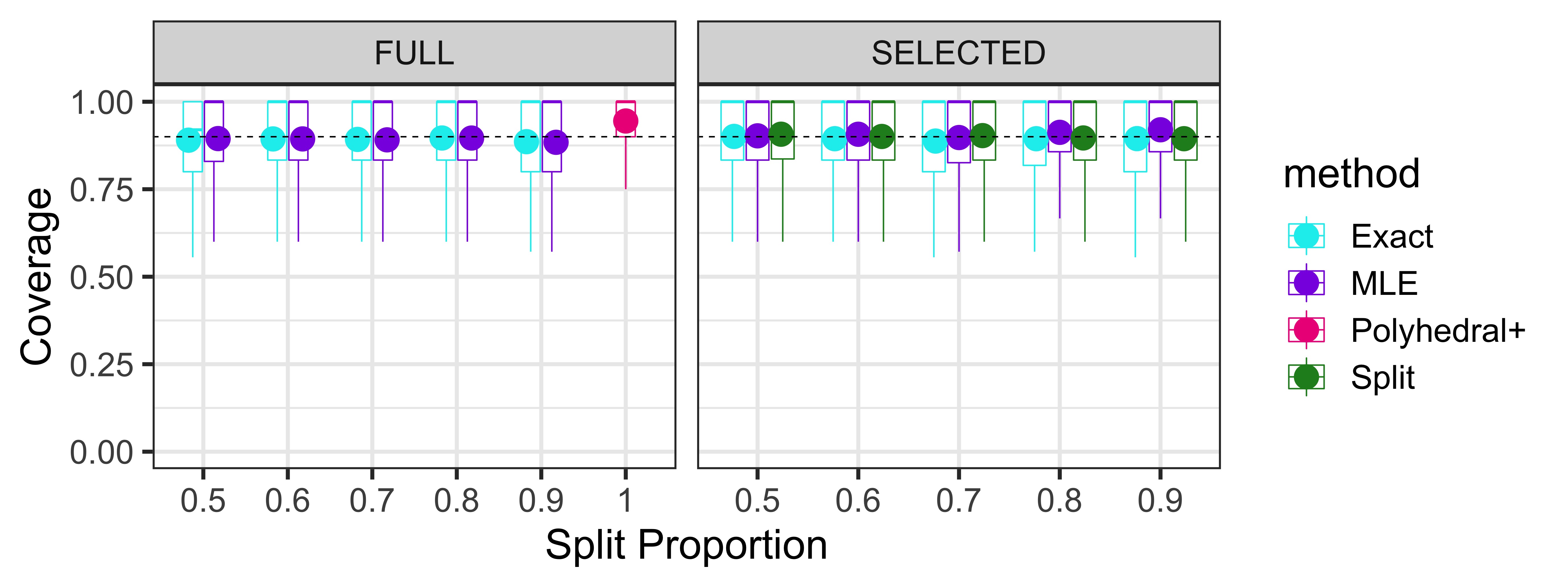}}
\end{center}
\vspace{-1cm}
\caption{Coverage rate of confidence intervals. Under Signal regime 3, left panel and right panel show distribution of coverage rates and the mean coverage over all $500$ replications in the Full and Selected Models, respectively.}
\label{fig:2}
\end{figure}

\begin{figure}[H]
\begin{center}
\centerline{\includegraphics[width=1.\linewidth]{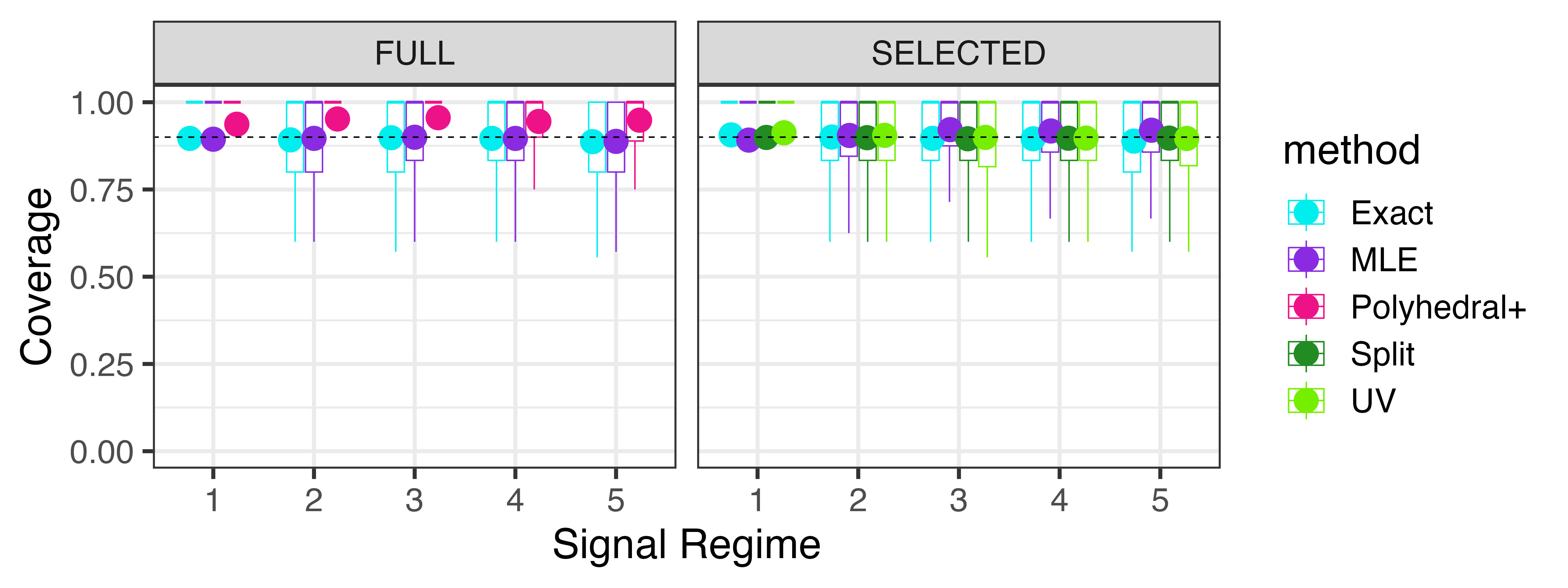}}
\end{center}
\vspace{-1cm}
\caption{Coverage rate of confidence intervals. At fixed split proportion $0.80$, left panel and right panel show distribution of coverage rates and the mean coverage over all $500$ replications in the Full and Selected Models, respectively.}
\label{fig:3}
\end{figure}

In Figures \ref{fig:4} and \ref{fig:5}, we investigate how the ``Exact" confidence intervals compare in length when we vary the split proportion and the strength of signals.

Under the Full Model, we observe that the interval lengths produced by ``Exact" and ``MLE" are consistently less variable than ``Polyhedral+". 
This observation is also true if we focus attention on split proportion $\rho=0.80$, at which the randomized methods are comparable with ``Polyhedral+" in terms of the quality of feature selection.

Similar patterns are seen in Figure \ref{fig:5} as we change the signal strengths under Signal Regimes 1-5. 
Under both models, we note that our ``Exact" method yields only nominally longer intervals than ``MLE", but, consistently gives shorter intervals than the two split-based strategies ``Split" and ``UV".
As previously mentioned, we only display the lengths of split-based methods for the Selected Model, as they are much longer than the other methods when used under the Full Model.
The increasing cost of discarding data from the selection stage is evident from the right panel of Figure \ref{fig:4}.

\begin{figure}[H]
\begin{center}
\centerline{\includegraphics[width=1.\linewidth]{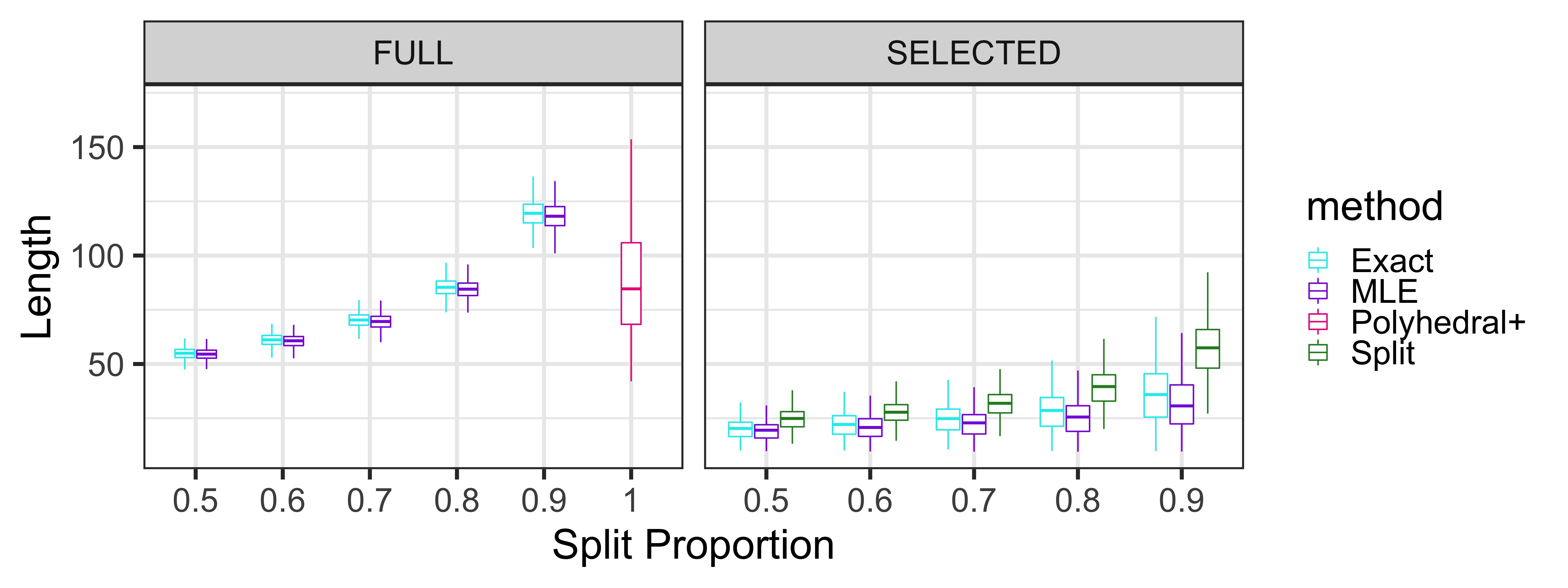}}
\end{center}
\vspace{-1cm}
\caption{Length of confidence intervals. Under Signal regime 3, left panel and right panel show distribution of lengths of confidence intervals over all $500$ replications in the Full and Selected Models, respectively.}
\label{fig:4}
\end{figure}

\begin{figure}[H]
\begin{center}
\centerline{\includegraphics[width=1.\linewidth]{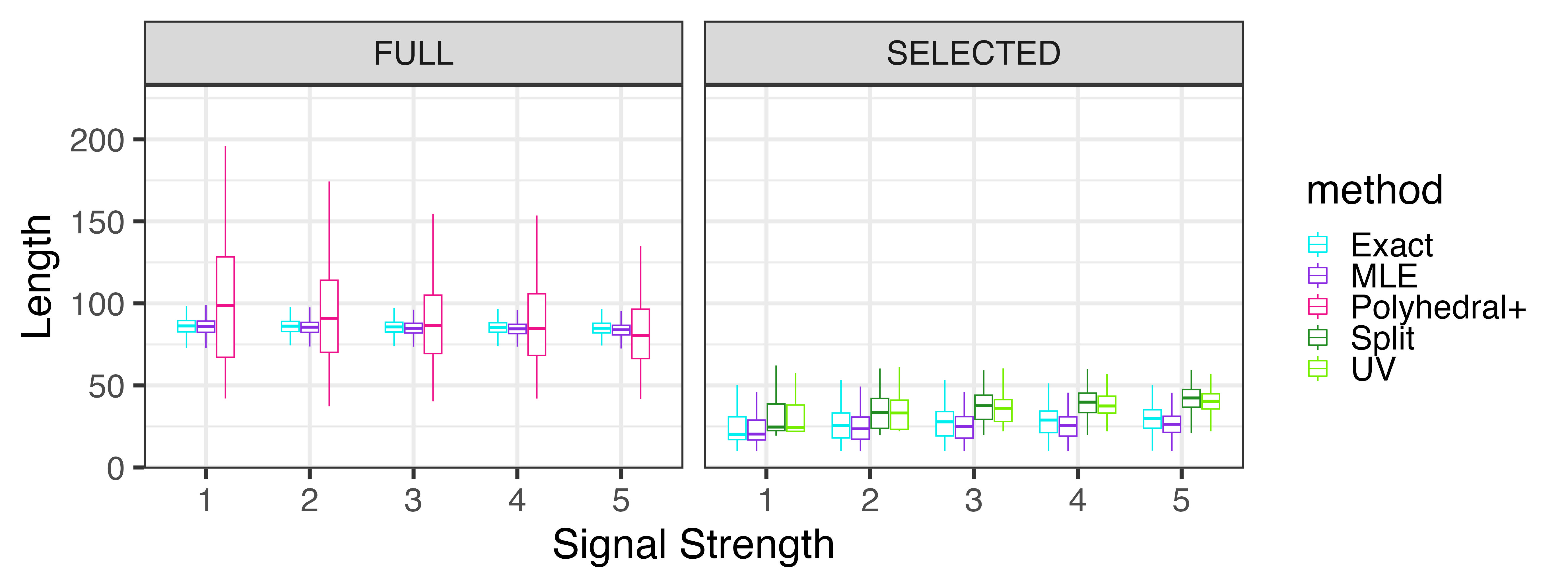}}
\end{center}
\vspace{-1cm}
\caption{Length of confidence intervals. At fixed split proportion $0.80$, left panel and right panel show distribution of lengths of confidence intervals over all $500$ replications in the Full and Selected Models, respectively.}
\label{fig:5}
\end{figure}

\section{Analysis of HIV drug resistance data}
\label{sec:real}

We apply our method to the HIV drug resistance data.
This dataset, originally analyzed by \cite{rhee2006genotypic}, is publicly available on the Stanford HIV Database (HIVDB).
The goal of the analysis is to find associations between mutations of the HIV virus and drug resistance to antiretroviral drugs.
We extract a part of this dataset that focuses on the response to one particular drug, Lamivudine (3TC), as has been described previously by \cite{bi2020inferactive, selective_bayesian}. 
The predictive features in this data are $91$ mutations that appeared more than $10$ times in the samples, and the response is a log-transformed value of the measurement for drug resistance.
Our dataset contains $633$ sample observations for the response and the set of $91$ features.

We focus on three randomized procedures for interval estimation.
To run our method, we consider drawing a Gaussian randomization variable $w \sim \mathcal{N}(0_p, \Omega)$, where $$\rho= \frac{n_1}{n}=0.8,$$ and $\Omega$ is set as per \eqref{randomization:carving}.
We implement the randomized LASSO with the randomization variable $w$.
The randomized LASSO selects a subset of $14$ mutations.
At the inference stage, we use our exact pivot to construct confidence intervals for the selected regression coefficients; our method is called ``Exact". 
For comparison, we construct approximate confidence intervals using ``MLE" after the same run of the randomized LASSO.
We also consider the intervals produced by ``Split" based on $\rho=0.8$. That is, ``Split" uses $80\%$ of the data samples for selecting features, and this resulted in selecting a subset of $17$ features. 
The remaining $20\%$ of the samples were reserved for selective inference.

Figure \ref{fig:hiv:2} depicts interval estimates produced by ``Exact", ``MLE" and ``Split".
The set of selected features is depicted on the x-axis.
To allow convenient visualization, the plot does not include the selected feature ‘P184V’, which has a different scale from the other variables in the selected set.
We note that ``Split" selected three features, ``P118I", ``P41L", ``P77L", that were not selected with the randomized LASSO.
Similarly, the randomized LASSO selected the mutation ``P69D" that was not picked by ``Split" at the selection stage.
But, these mutations were not significant after selective inference was conducted with the methods associated with each case.
At the stage of inference, we present confidence intervals for a feature given it was selected in our model.

We observe that the interval estimators for the selected effects produced by the three methods are in close agreement with most features.
Overall, we note that the two randomized methods which reuse data from the selection stage seem to find a larger set of significant associations.
On an average, the length of interval estimators based on ``Exact" is equal to $3.76$.
In agreement with our simulated results, the ``Exact" intervals are longer than the ``MLE" intervals which have an average length of $2.76$. 
This is the price that we pay in exchange of exact selective inference with our pivot.
Our intervals are, however, much shorter than the related ``Split" procedure; the average length of intervals produced by ``Split" in this instance is equal to $6.58$.
Figure \ref{fig:hiv:3} displays box plots for the lengths of the intervals, which exhibit this pattern.

\begin{figure}[H]
\begin{center}
\centerline{\includegraphics[width=18cm, height=9cm]{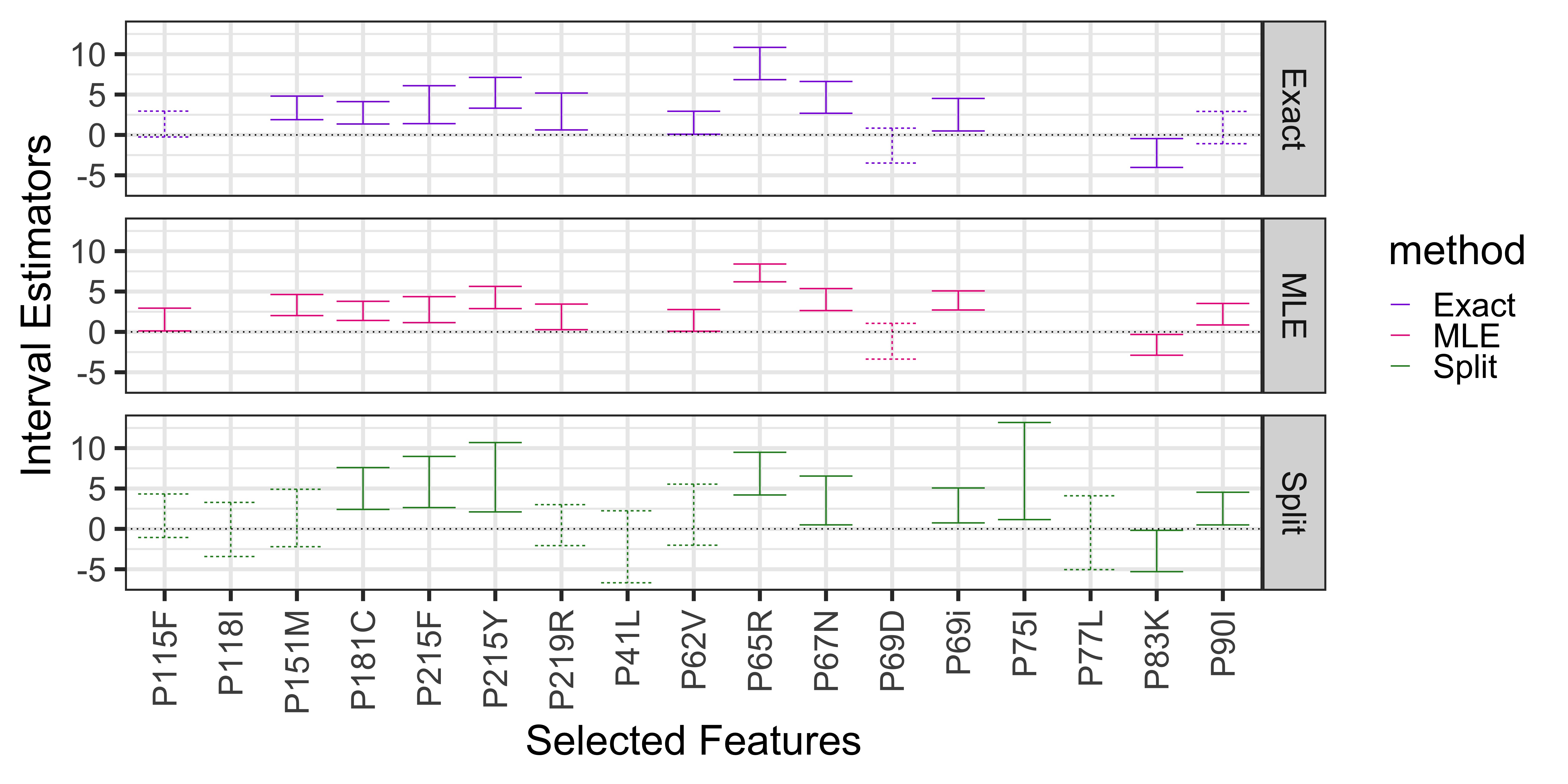}}
\end{center}
\vspace{-1cm}
\caption{Confidence intervals based on the three randomized methods: ``Exact", ``MLE" and ``Split". Solid lines are used for interval estimators that do not cover $0$. Dotted lines are used for interval estimators that cover $0$.}
\label{fig:hiv:2}
\end{figure}

\begin{figure}[H]
\begin{center}
\centerline{\includegraphics[width=1.\linewidth]{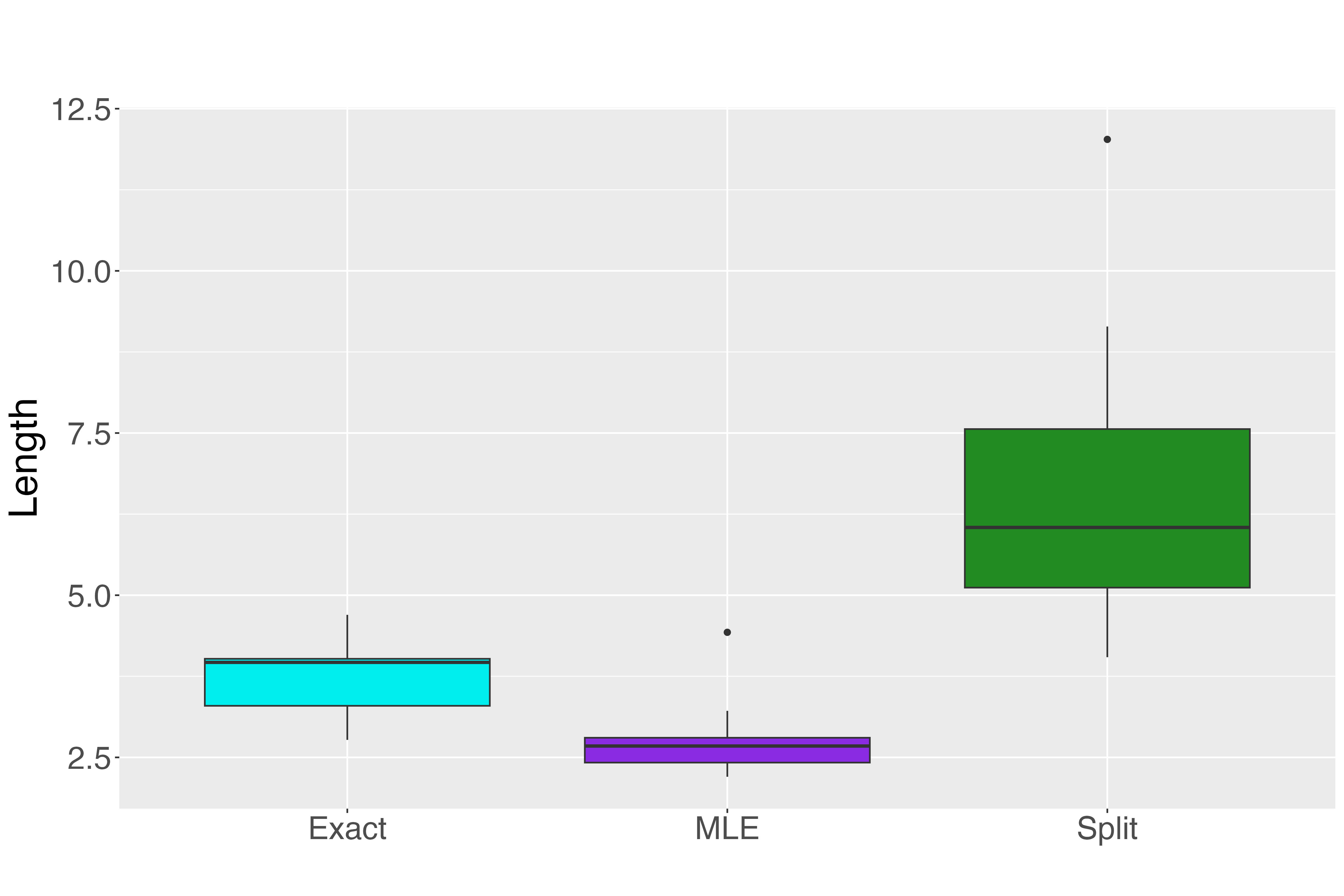}}
\end{center}
\vspace{-1cm}
\caption{Lengths of interval estimators.}
\label{fig:hiv:3}
\end{figure} 

\section{Discussion}
\label{sec:conclusion}

When conducting selective inference, randomizing data at the time of selection and then conditioning on the outcome of selection can significantly decrease the length of confidence intervals. 
However, the pivot used in previous methods is often unavailable in closed form. 
Our paper introduces a new pivot for simple Gaussian randomization schemes that is easy to compute and eliminates the need for any further approximation. 

Although exact selective inference has its benefits, it also comes at a cost. 
By conditioning on additional information to obtain our pivot, we sacrifice some power when compared to approximate techniques developed in prior work, such as \cite{panigrahi2017mcmc, panigrahi2019approximate}.
Our simulated findings for popular Gaussian regression models show that the loss in power with our method is nominal for the well-studied LASSO algorithm. 
In fact, with a carefully chosen randomization scheme, our confidence intervals can be much shorter than those produced by data splitting. 
The gains from reusing data from the selection stages become more pronounced as the number of samples available for inference decreases. 
Therefore, our method can be applied to carry out selective inference when the dataset is not large enough to be split into two parts or when there is no simple way to divide the dataset into independent subsamples.

The focus of this paper has been on exact inferences in the least squares problem. 
However, the pivot generated from the randomization scheme used in the paper can also be applied to more general estimation problems, including the class of M-estimation problems. We believe that the same pivot could be used as long as the selection algorithm allows for a linear representation in optimization variables at the solution. 
In these problems, our pivot would provide asymptotic selective inferences instead of exact selective inferences, which would require a formal theoretical justification and needs to be investigated in future work.

\section{Acknowledgements}
S. Panigrahi's research is supported in part by NSF grants: DMS 1951980 and DMS 2113342.
K. Fry's research is supported by NSF GRFP. 
J. Taylor's research is supported by ARO grant: 70940MA.

\bibliographystyle{apalike}  
\bibliography{references}

\newpage

\section{Appendix}

\subsection{Proofs of technical results}
\begin{proof} Proposition \ref{Lem:cond}.
We begin by writing
$$
O =  (r^{ j})^{\tp} O \frac{\Theta r^{j}}{(r^{j})^{\tp} \Theta r^{j}} + A^{r^j}= (r^{j})^{\tp} O  \cdot Q^{j} + A^{r^j}.
$$
Then, using \eqref{cond:event}, we have
\begin{equation*}
    \begin{aligned}
    &\left\{ G= \cG, A^{r^j}= \cA^{r^j} \right\} \\
    &= \left\{ LO< M, U=\cU, A^{r^j}= \cA^{r^j}\right\}\\
    &= \left\{  (r^{ j})^{\tp} O \cdot LQ^{j} < M - L\cA^{r^j}, U=\cU, A^{r^j}= \cA^{r^j} \right\}\\
    &= \left\{I^j_{-} <  (r^{j})^{\tp} O < I^j_{+}, \; U=\cU, A^{r^j}= \cA^{r^j} \right\}.
    \end{aligned}
\end{equation*}
\end{proof}

Suppose that $\cE \subset [p]$ and $\cS \in \{-1,1\}^{|\cE|}$ are a fixed set and a fixed set of signs, respectively.
Before providing a proof for Theorem \ref{thm:main}, we state a few results on the distribution of our optimization variables given $Y=y$.

\begin{lemma}
Define 
$$\Pi_y(\cO, \cU)=  Py+ Q\cO + R \cU + T.$$
The joint density of $O$, $U$ given $Y=y$,at $(\cO,\cU)$,  is proportional to
$$
\phi\left(\Pi_y\left(\cO, \cU\right); 0_p, \Omega\right),
$$
and the density of $O$ given $U=\cU$ and $Y=y$, at $\cO$, is equal to
$$
\phi\left(\cO; \Delta(y, \cU), \Theta \right).
$$
\label{lem:cond:density}
\end{lemma}

\begin{proof} Lemma \ref{lem:cond:density}.
Note that the density of our randomization variable $W$ given $Y=y$ is equal to
$$
\phi(w; 0, \Omega)
$$
at $w$.
To derive the density for the optimization variables, we use the following change of variables 
$$
W \to (O,U), \text{ where } (O,U) = \Pi_y^{-1}(W).
$$
Then, the density of the new variables $O$ and $U$, at $(\cO,\cU)$, is given by
$$
J\cdot \phi(\Pi_y(\cO,\cU); 0_p, \Omega), 
$$
where
$$J= \Big|\text{det}\begin{bmatrix} Q & R \end{bmatrix}\Big|$$
is the Jacobian associated with the change of variables through $\Pi_y$. This proves the first part of our claim.

Next, we observe that the conditional density of $O$ given $U=\cU$ and $Y=y$, at $\cO$, is equal to
$$
\dfrac{J\cdot \phi(\Pi_y(\cO,\cU); 0_p, \Omega)}{\int J\cdot \phi(\Pi_y(o,\cU); 0_p, \Omega) do}= \phi(\Delta(y, \cU), \Theta).
$$
\end{proof}

\begin{lemma}
The two variables $(r^{j})^{\tp} O$ and $A^{r^j}$ are independent given $Y=y$, $U=\cU$.
\label{lem:indep}
\end{lemma}

\begin{proof}
This claim follows directly by using the fact that the covariance of $O\lvert Y=y, U=\cU$ is $\Theta$, as derived in Lemma \ref{lem:cond:density}.
Now, we observe
$$\text{Cov}(A^{r^j}, (r^{j})^{\tp} O \lvert Y=y, U=\cU)=0_{|\cE|}.$$
\end{proof}

Now, we are ready to derive the pivot in Theorem \ref{thm:main}.
In our proof, we use the symbols $\ell_V(v)$ and $\ell_{V|X}(v| x)$ for the density of a variable $V$ and the conditional density of a variable $V$ given $X=x$, at $v$, respectively.
In particular, if the density functions involve our parameter of interest, $\beta^{\cE}_j$, we indicate this through the symbols $\ell_{V;\beta^{\cE}_j}(v)$ and $\ell_{V|X;\beta^{\cE}_j}(v| x)$.

\begin{proof} 
Theorem \ref{thm:main}. \ We divide our proof into three steps.

In Step 1, we start from the marginal density of 
$$(\h{\beta}_j^{\cE}, \h{\Gamma}^j, (r^{j})^{\tp} O, A^{r^j},U)$$ for the fixed set $\cE$ and fixed signs $\cS$.
This marginal density, at $(b, g, \rZ, \cA^{r^j}, \cU)$, is equal to
\begin{equation*}
    \begin{aligned}
& \ell_{\h{\beta}_j^{\cE};\beta_j^{\cE}}(b)\cdot \ell_{\h{\Gamma}^j}(g) \cdot \ell_{(r^{j})^{\tp} O, A^{r^j},U\lvert \h{\beta}_j^{\cE},\h{\Gamma}^j}( \rZ,\cA^{r^j}, \cU \lvert b, g)
 \end{aligned}
\end{equation*}
which we obtain by using the fact that
$$ 
y = \frac{c^{j}}{\|c^{j}\|^2_2} \h{\beta}_j^{\cE} + \h{\Gamma}^j= V^j\h{\beta}_j^{\cE} + \h{\Gamma}^j,
$$
and that $\h{\beta}_j^{\cE}$ and $\h{\Gamma}^j$ are independent variables.
The above-stated density simplifies as
\begin{equation*}
    \begin{aligned}
&\ell_{\h{\beta}_j^{\cE};\beta_j^{\cE}}(b)\cdot \ell_{\h{\Gamma}^j}(g) \cdot \ell_{(r^{j})^{\tp} O, A^{r^j}\lvert \h{\beta}_j^{\cE},\h{\Gamma}^j, \cU}( \rZ, \cA^{r^j} \lvert b, g, \cU)\cdot \ell_{U\lvert \h{\beta}_j^{\cE},\h{\Gamma}^j}( \cU \lvert b, g)\\
&=\ell_{\h{\beta}_j^{\cE};\beta_j^{\cE}}(b)\cdot \ell_{\h{\Gamma}^j}(g) \cdot \ell_{(r^{j})^{\tp} O\lvert \h{\beta}_j^{\cE},\h{\Gamma}^j, \cU}(  \rZ\lvert b, g, \cU) \cdot \ell_{A^{r^j}\lvert \h{\beta}_j^{\cE},\h{\Gamma}^j, \cU}( \cA^{r^j} \lvert b, g, \cU)\cdot \ell_{U\lvert \h{\beta}_j^{\cE},\h{\Gamma}^j}( \cU \lvert b, g),
 \end{aligned}
\end{equation*}
The expression on the right-hand side follows by using the conditional independence between $(r^{j})^{\tp} O$ and $A^{r^j}$ which was shown in Lemma \ref{lem:indep}.

In Step 2, we derive the density of $\h{\beta}_j^{\cE}$ and $(r^j)^{\tp}O$ when conditioned on the event
$$\left\{ G= \cG, A^{r^j}= \cA^{r^j} \right\}$$
and the value of $\h{\Gamma}^j$.
Let $C( \h{\Gamma}^j, \cU, A^{r^j})$ be equal to
$$\int \ell_{\h{\beta}_j^{\cE};\beta_j^{\cE}}(\tilde{b}) \cdot \ell_{(r^{j})^{\tp} O\lvert \h{\beta}_j^{\cE},\h{\Gamma}^j, \cU}( \tilde{\rZ}\lvert \tilde{b}, \h{\Gamma}^j, \cU) \cdot \ell_{A^{r^j}\lvert \h{\beta}_j^{\cE},\h{\Gamma}^j, \cU}( \cA^{r^j} \lvert \tilde{b}, \h{\Gamma}^j, \cU)\cdot \ell_{U\lvert \h{\beta}_j^{\cE},\h{\Gamma}^j}( \cU \lvert \tilde{b}, \h{\Gamma}^j) \cdot 1_{[I^j_{-}, I^j_{+}]}(\tilde{\rZ})d\tilde{\rZ} d\tilde{b}.$$
Because of the characterization for our conditioning event in Proposition \ref{Lem:cond}, the conditional density of $(\h{\beta}_j^{\cE}, r^j)^{\tp}O)$ at $(b, \rZ)$ is equal to
\begin{equation*}
    \begin{aligned}
    &(C( \h{\Gamma}^j, \cU, A^{r^j}))^{-1}\cdot \ell_{\h{\beta}_j^{\cE};\beta_j^{\cE}}(b)\cdot \ell_{(r^{j})^{\tp} O\lvert \h{\beta}_j^{\cE},\h{\Gamma}^j, \cU}(  \rZ\lvert b, \h{\Gamma}^j, \cU) \cdot \ell_{A^{r^j}\lvert \h{\beta}_j^{\cE},\h{\Gamma}^j, \cU}( \cA^{r^j} \lvert b, \h{\Gamma}^j, \cU)\\
    &\;\;\;\;\;\;\;\;\;\;\;\;\;\;\;\;\;\;\;\;\;\;\;\;\;\;\;\;\;\;\;\;\;\;\;\;\;\;\;\;\;\;\;\;\;\;\;\;\;\;\;\;\;\;\;\;\;\;\;\;\;\;\;\;\;\;\;\;\;\;\;\;\;\;\;\;\;\;\;\;\;\;\;\;\;\;\;\;\;\;\;\;\;\times \ell_{U\lvert \h{\beta}_j^{\cE},\h{\Gamma}^j}( \cU \lvert b, \h{\Gamma}^j) \cdot 1_{[I^j_{-}, I^j_{+}]}(\rZ).
    \end{aligned}
\end{equation*}

In Step 3, we simplify the conditional density from the preceding step.
Observe that
\begin{equation}
    \begin{aligned}
    & \ell_{A^{r^j}\lvert \h{\beta}_j^{\cE},\h{\Gamma}^j, \cU}( \cA^{r^j} \lvert b, \h{\Gamma}^j, \cU) =   \ell_{A^{r^j}\lvert \h{\Gamma}^j, \cU}( \cA^{r^j} \lvert \h{\Gamma}^j, \cU).
    \end{aligned}
        \label{simplify:1}
\end{equation}
This is because the conditional Gaussian distribution of $A^{r^j}$ on the left-hand side display depends on $\h{\beta}_j^{\cE}$ only through its mean, which is equal to
\begin{equation*}
    \begin{aligned}
\left(  Q^j(r^{j})^{\tp}-I \right)\Theta Q^{\tp}\Omega^{-1} (P^{j}b + P\h{\Gamma}^j + R\cU + T)&= \left( I-  Q^j(r^{j})^{\tp} \right)\Delta(\h{\Gamma}^j , \cU).
    \end{aligned}
\end{equation*}
Note that the expression on the right-hand side is free of $b$.
Thus, we can further write the conditional density of $\h{\beta}_j^{\cE}$ and $(r^j)^{\tp}O$ as 
\begin{equation*}
    \begin{aligned}
    &\dfrac{\ell_{\h{\beta}_j^{\cE};\beta_j^{\cE}}(b)\cdot \ell_{(r^{j})^{\tp} O\lvert \h{\beta}_j^{\cE},\h{\Gamma}^j, \cU}(  \rZ\lvert b, \h{\Gamma}^j, \cU) \cdot \ell_{U\lvert \h{\beta}_j^{\cE},\h{\Gamma}^j}( \cU \lvert b, \h{\Gamma}^j)}{\int \ell_{\h{\beta}_j^{\cE};\beta_j^{\cE}}(\tilde{b}) \cdot \ell_{(r^{j})^{\tp} O\lvert \h{\beta}_j^{\cE},\h{\Gamma}^j, \cU}( \tilde{\rZ}\lvert \tilde{b}, \h{\Gamma}^j, \cU) \cdot \ell_{U\lvert \h{\beta}_j^{\cE},\h{\Gamma}^j}( \cU \lvert \tilde{b}, \h{\Gamma}^j) \cdot 1_{[I^j_{-}, I^j_{+}]}(\tilde{\rZ})d\tilde{\rZ} d\tilde{b} } \cdot 1_{[I^j_{-}, I^j_{+}]}(\rZ)\\
    &=  \dfrac{ \phi(b; \lambda^j\beta_j^{\cE}+ \zeta^j, (\sigma^j)^2) \cdot \phi(\rZ; \theta^j(b),(\vartheta^j)^2)}{\int \phi(\tilde{b}; \lambda^j\beta_j^{\cE}+ \zeta^j, (\sigma^j)^2) \cdot \phi(\tilde{\rZ}; \theta^j(\tilde{b}),(\vartheta^j)^2) \cdot 1_{[I^j_{-}, I^j_{+}]}(\tilde{\rZ})d\tilde{\rZ} d\tilde{b}}\cdot 1_{[I^j_{-}, I^j_{+}]}(\rZ),
    \end{aligned}
\end{equation*}
because
\begin{equation*}
    \begin{aligned}
    & \ell_{\h{\beta}_j^{\cE};\beta_j^{\cE}}(b) \cdot \ell_{(r^{j})^{\tp} O\lvert \h{\beta}_j^{\cE},\h{\Gamma}^j, \cU}(  \rZ\lvert b, \h{\Gamma}^j, \cU) \cdot \ell_{U\lvert \h{\beta}_j^{\cE},\h{\Gamma}^j}( \cU \lvert b, \h{\Gamma}^j)   \propto \phi(b; \lambda^j\beta_j^{\cE}+ \zeta^j, (\sigma^j)^2)\cdot \phi(\rZ; \theta^j(b),(\vartheta^j)^2).
    \end{aligned}
\end{equation*}

Marginalizing over $(r^j)^{\tp}O$ yields us the following conditional density
\begin{equation*}
    \begin{aligned}
    \dfrac{\phi\left(b; \lambda^j\beta^{\cE}_j+ \zeta^j, (\sigma^j)^2\right)\cdot \Phi\left(\frac{1}{\vartheta^j}(I^j_{+}- \theta^j(b))\right)- \Phi\left(\frac{1}{\vartheta^j}(I^j_{-}- \theta^j(b))\right)}{\int \phi\left(\tilde{b}; \lambda^j\beta^{\cE}_j+ \zeta^j, (\sigma^j)^2\right)\cdot  \Phi\left(\frac{1}{\vartheta^j}(I^j_{+}- \theta^j(\tilde{b}))\right)- \Phi\left(\frac{1}{\vartheta^j}(I^j_{-}- \theta^j(\tilde{b}))\right) d\tilde{b} }.
    \end{aligned}
\end{equation*}
A probability integral transform of the related conditional distribution gives us $\mathcal{P}^j_{\text{\normalfont Exact}} (\beta^{\cE}_j)$, a uniformly distributed variable on $[0,1]$.
\end{proof}

\begin{proof}
Corollary \ref{exact:pivot:carving}.
The proof of this corollary follows by noting that 
\begin{align*}
\begin{gathered}
\lambda^j=1, \zeta^j=0, \\
(\sigma^j)^2 = \sigma^2 \|c^j\|_2^2,
\end{gathered}
\end{align*}
when $\Omega$ is set according to \eqref{randomization:carving}.
Additionally, we have
$$r^j = -\frac{1}{\tau^2\|c^j\|_2^2}e_j, \text{ and } \  \Theta Q^{\tp}\Omega^{-1}= \begin{bmatrix}  (X_{\cE}^{\tp} X_{\cE})^{-1} & 0_{|\cE|, p-|\cE|}\end{bmatrix},$$
which leads us to the claimed values for  $\vartheta^j$ and $\theta^j(x)$.
\end{proof}

\begin{proof}
Proposition \ref{choice:cond}. 
We continue with the notations in the proof of Theorem \ref{thm:main}.
For fixed set $\cE\subset [p]$ and fixed signs $\cS\in \{-1,1\}^{|\cE|}$, we establish the stronger assertion that
\begin{equation}
\ell_{\h{\beta}_j^{\cE} \lvert U, A^{r^j}, \h{\Gamma}^j; \beta^{\cE}_j}(b \lvert \cU, \cA^{r^j}, g)=\ell_{\h{\beta}_j^{\cE} \lvert U, \h{\Gamma}^j; \beta^{\cE}_j}(b \lvert \cU, g).
\label{stronger:statement}
\end{equation}
Starting with the distribution related to the conditional density on the left-hand side of \eqref{stronger:statement}, we have
\begin{equation*}
    \begin{aligned}
    &\ell_{\h{\beta}_j^{\cE} \lvert U, A^{r^j}, \h{\Gamma}^j; \beta^{\cE}_j}(b \lvert \cU, \cA^{r^j}, g)\\
    &=\dfrac{\ell_{\h{\beta}_j^{\cE};\beta_j^{\cE}}(b)\cdot \ell_{A^{r^j}\lvert \h{\beta}_j^{\cE},\h{\Gamma}^j, \cU}( \cA^{r^j} \lvert b, g, \cU)\cdot \ell_{U\lvert \h{\beta}_j^{\cE},\h{\Gamma}^j}( \cU \lvert b, g)\cdot \bigintsss\ell_{(r^{j})^{\tp} O\lvert \h{\beta}_j^{\cE},\h{\Gamma}^j, \cU}(  \tilde{\rZ}\lvert b, g, \cU) d\tilde{\rZ} }{\bigintsss \ell_{\h{\beta}_j^{\cE};\beta_j^{\cE}}(\tilde{b}) \cdot \ell_{A^{r^j}\lvert \h{\beta}_j^{\cE},\h{\Gamma}^j, \cU}( \cA^{r^j} \lvert \tilde{b}, g, \cU)\cdot \ell_{U\lvert \h{\beta}_j^{\cE},\h{\Gamma}^j}( \cU \lvert \tilde{b}, g)  \cdot \ell_{(r^{j})^{\tp} O\lvert \h{\beta}_j^{\cE},\h{\Gamma}^j, \cU}( \tilde{\rZ}\lvert \tilde{b}, g, \cU)d\tilde{\rZ} d\tilde{b} } \\
    &= \dfrac{\ell_{\h{\beta}_j^{\cE};\beta_j^{\cE}}(b)\cdot \ell_{A^{r^j}\lvert \h{\Gamma}^j, \cU}( \cA^{r^j} \lvert g, \cU)\cdot \ell_{U\lvert \h{\beta}_j^{\cE},\h{\Gamma}^j}( \cU \lvert b, g)\cdot \bigintsss\ell_{(r^{j})^{\tp} O\lvert \h{\beta}_j^{\cE},\h{\Gamma}^j, \cU}(  \tilde{\rZ}\lvert b, g, \cU) d\tilde{\rZ} }{\bigintsss \ell_{\h{\beta}_j^{\cE};\beta_j^{\cE}}(\tilde{b}) \cdot \ell_{A^{r^j}\lvert \h{\Gamma}^j, \cU}( \cA^{r^j} \lvert g, \cU)\cdot \ell_{U\lvert \h{\beta}_j^{\cE},\h{\Gamma}^j}( \cU \lvert \tilde{b}, g)  \cdot \ell_{(r^{j})^{\tp} O\lvert \h{\beta}_j^{\cE},\h{\Gamma}^j, \cU}( \tilde{\rZ}\lvert \tilde{b}, g, \cU)d\tilde{\rZ} d\tilde{b} } \\  
    &= \dfrac{\ell_{\h{\beta}_j^{\cE};\beta_j^{\cE}}(b)\cdot \ell_{U\lvert \h{\beta}_j^{\cE},\h{\Gamma}^j}( \cU \lvert b, g)\cdot \bigintsss\ell_{(r^{j})^{\tp} O\lvert \h{\beta}_j^{\cE},\h{\Gamma}^j, \cU}(  \tilde{\rZ}\lvert b, g, \cU) d\tilde{\rZ} }{\bigintsss \ell_{\h{\beta}_j^{\cE};\beta_j^{\cE}}(\tilde{b})\cdot \ell_{U\lvert \h{\beta}_j^{\cE},\h{\Gamma}^j}( \cU \lvert \tilde{b}, g)  \cdot \ell_{(r^{j})^{\tp} O\lvert \h{\beta}_j^{\cE},\h{\Gamma}^j, \cU}( \tilde{\rZ}\lvert \tilde{b}, g, \cU)d\tilde{\rZ} d\tilde{b} }.
 \end{aligned}
\end{equation*}
Note, to derive the above expression, we used the fact in \eqref{simplify:1}.
Clearly, the expression for the conditional density does not depend on $\cA^{r^j}$.
Hence, we have proved our assertion and conclude that
\begin{equation*}
    \begin{aligned}
\text{\normalfont Var}\left(\h{\beta}_j^{\cE}\ \Big\lvert \ U= \cU, A^{r^j}= \cA^{r^j}, \h{\Gamma}^j= g\right) &= \text{\normalfont Var}\left(\h{\beta}_j^{\cE} \ \Big\lvert \ U= \cU, \h{\Gamma}^j= g\right)\\
&= \underset{\eta}{\text{maximize}} \ \text{\normalfont Var}\left(\h{\beta}_j^{\cE}\ \Big\lvert \ U= \cU, A^{\eta}= \cA^{\eta}, \h{\Gamma}^j=g\right).
\end{aligned}
\end{equation*}
\end{proof}

\subsection{Connection with data splitting}
\label{sec:split}

We provide additional details to connect the randomized LASSO with data splitting.
Denote by $S\subset [n]$ a random subsample of size $n_1$.
Let $X^{(S)}$ and $y^{(S)}$ denote the feature matrix and the response vector which contain the observations in this subsample $S$.
In this discussion, we emphasize the dependence on the sample size, denoted by $n$, whenever it is relevant.

Fix $\lambda \in \mathbb{R}^+$. 
Let 
$$\rho= \frac{n_1}{n}.$$
Akin to data splitting, suppose we solve the LASSO using the data $(y^{(S)}, X^{(S)})$ as
\begin{equation}
   \underset{b \in \real^p}{\text{minimize}}\ \frac{1}{2\sqrt{n_1}}\|y^{(S)}-X^{(S)}b\|_2^2 + \lambda^{(S)} \|b\|_1,
\label{split:LASSO:eg}   
\end{equation}
where 
$$ \lambda^{(S)}= \sqrt{\rho} \lambda.$$ 

At the solution of \eqref{split:LASSO:eg}, which we denote by $\widehat{b}^{(S)}$, observe that
\begin{equation}
\label{rand:asymptotic}
0_p = \lambda \begin{pmatrix} \cS \\ \cU \end{pmatrix} - \frac{1}{\rho\sqrt{n}} (X^{(S)})^\tp (y - X^{(S)} \widehat{b}^{(S)}).
\end{equation}
Here the first term on the right-hand side of the display is the subgradient of the LASSO penalty at the solution.
Let $E=\cE$ denote the selected set of features.

Define
\begin{equation*}
\begin{aligned}
w &= \frac{1}{\rho\sqrt{n}} (X^{(S)})^\tp (y - X^{(S)} \widehat{b}^{(S)})  - \frac{1}{\sqrt{n}} X^\tp (y - X \widehat{b}^{(S)})\\
&= \sqrt{n}\left\{\frac{1}{n_1} (X^{(S)})^\tp (y - X^{(S)} \widehat{b}^{(S)})  - \frac{1}{n} X^\tp (y - X \widehat{b}^{(S)}) \right\}
\end{aligned}
\end{equation*}
as our randomization variable.
Note that the equality in \eqref{rand:asymptotic} can be rewritten as 
\begin{equation*}
\begin{aligned}
 \frac{1}{\rho\sqrt{n}} (X^{(S)})^\tp (y - X^{(S)} \widehat{b}^{(S)})  - \frac{1}{\sqrt{n}} X^\tp (y - X \widehat{b}^{(S)})   =   \lambda \begin{pmatrix} \cS \\ \cU \end{pmatrix} - \frac{1}{\sqrt{n}} X^\tp (y - X \widehat{b}^{(S)}),
\end{aligned}
\end{equation*}
which is equivalent to
$$w  =\lambda \begin{pmatrix} \cS \\ \cU \end{pmatrix} - \frac{1}{\sqrt{n}} X^\tp (y - X_{\cE} \cO).$$
As shown in Section \ref{sec3.1}, \eqref{rand:asymptotic} can be expressed as a linear mapping in the optimization variables $O$ and $U$. 
Moreover, $w_n$ can be proven to asymptotically follow a Gaussian distribution with the variance matrix specified in \eqref{randomization:carving}. 
To see a formal derivation of the asymptotic distribution of $w_n$, we refer readers to \cite{selective_bayesian}.
\end{document}